\documentclass[a4paper,UKenglish,numberwithinsect]{lipics-v2021}

\usepackage{booktabs}
\usepackage{float}
\usepackage{tikz}
\usetikzlibrary{arrows.meta,positioning}
\usepackage{pgfplots}
\usepgfplotslibrary{groupplots}
\pgfplotsset{compat=1.18}

\definecolor{staircellblue}{RGB}{216,233,247}
\definecolor{stairlineblue}{RGB}{86,139,173}
\definecolor{stairblockgreen}{RGB}{220,244,224}
\definecolor{stairgreenline}{RGB}{38,145,71}
\definecolor{stairblockgray}{RGB}{238,240,242}
\definecolor{stairgrayline}{RGB}{112,120,128}
\definecolor{stairwarmline}{RGB}{230,126,34}
\definecolor{stairredline}{RGB}{201,55,48}
\definecolor{stairtextgray}{RGB}{82,91,99}
\definecolor{benchmarkproposed}{RGB}{215,48,39}
\definecolor{benchmarkbaseline}{RGB}{117,81,166}
\definecolor{benchmarkgrid}{RGB}{205,210,215}

\hideLIPIcs

\floatstyle{plain}
\newfloat{algorithm}{tbp}{loa}
\floatname{algorithm}{Algorithm}
\newcounter{algline}
\newenvironment{algorithmic}{%
  \setcounter{algline}{0}%
  \begin{tabular}{@{}r@{\hspace{0.7em}}p{\dimexpr\linewidth-3em\relax}@{}}%
}{\end{tabular}}
\newcommand{\AlgLine}[2]{%
  \stepcounter{algline}%
  \makebox[1.5em][r]{\scriptsize\thealgline} & \hspace*{#1em}#2\\%
}

\newcommand{\floor}[1]{\left\lfloor #1\right\rfloor}

\newcommand{\coeff}[2]{[z^{#1}]#2}
\newcommand{\M}{\mathsf M}
\newcommand{\MZ}{\mathsf M_{\mathbb Z}}
\DeclareMathOperator{\mult}{\mathcal W}

\title{Computing All Lattice-Rectangle Counts by Rational Staircase Sums}
\titlerunning{Computing All Lattice-Rectangle Counts}

\author{Dmitry Babichev}{Independent researcher, France}{}{}{}
\author{Denis Pinchuk}{Independent researcher, USA}{}{}{}
\authorrunning{Dmitry Babichev and Denis Pinchuk}
\Copyright{Dmitry Babichev and Denis Pinchuk}

\ccsdesc[500]{Theory of computation~Design and analysis of algorithms}
\ccsdesc[300]{Theory of computation~Computational geometry}
\keywords{Lattice rectangles, exact counting, generating functions, product trees, symbolic computation}

\EventEditors{}
\EventNoEds{1}
\EventLongTitle{}
\EventShortTitle{}
\EventAcronym{}
\EventYear{}
\EventDate{}
\EventLocation{}
\EventLogo{}
\SeriesVolume{}
\ArticleNo{}

\providecommand{\reviewlinemode}{0}
\ifnum\reviewlinemode=1
  \modulolinenumbers[5]
\fi

\begin{document}
\ifnum\reviewlinemode=1
  \linenumbers
\else
  \nolinenumbers
\fi
\hypersetup{pdfsubject={Preprint}}
\raggedbottom
\maketitle
\makeatletter
\let\@oddfoot\@empty
\let\@evenfoot\@empty
\makeatother

\begin{abstract}
Let $F(n)$ be the number of rectangles, not necessarily axis-parallel, whose
vertices belong to the $n\times n$ square grid of lattice points.  We compute
 the complete table $F(1),\ldots,F(N)$ exactly in
 $O(\M(N)\log N)$ coefficient-ring operations and $O(N\log N)$ ring elements
 of working memory, where $\M(N)$ is a regular bound for multiplying degree-$N$
 polynomials.  The ring-level statement assumes that $6$ is invertible; over
 $\mathbb Z$ the only division is instead performed exactly in the elementary
 boundary term.  With quasi-linear polynomial multiplication the arithmetic
 bound is $O(N\log^2 N)$.  The algorithm applies a square-root cover before
coefficient extraction and evaluates the resulting rational wedge and
triangular sums by a local-denominator divide-and-conquer recursion.  Primitive
directions are recovered coefficientwise by M\"obius inversion, followed by
five prefix sums.  A modular number-theoretic-transform (NTT) implementation
with certified Chinese-remainder (CRT) recovery is evaluated experimentally
against the $O(N^{3/2})$ all-values algorithm.
\end{abstract}

\section{Introduction}

Let $F(n)$ denote the number of rectangles, axis-parallel or oblique, whose
four vertices belong to the point grid $\{0,\ldots,n-1\}^2$.  The sequence
$(F(n))_{n\geqslant1}$ is OEIS A085582~\cite{oeis}.  Babichev and
Babichev~\cite{paper} give an
$O(N^{3/2})$-operation algorithm for constructing the complete table
$F(1),\ldots,F(N)$.  We give an exact algorithm with arithmetic complexity
 $O(\M(N)\log N)$.  Throughout, $\M(N)$ is a nondecreasing, superadditive,
 at-least-linear bound for multiplying degree-$N$ polynomials over the
 coefficient ring, and $\M(cN)=O(\M(N))$ for every fixed $c>0$.  These standard
 regularity assumptions make constant changes in truncation length harmless.

 The problem lies at the intersection of lattice-point enumeration, primitive
 direction counting, and fast formal-series computation.  Barvinok's
 fixed-dimensional method encodes lattice points in rational polyhedra by short
 rational generating functions~\cite{barvinok}; its algorithmic development
 includes Barvinok and Pommersheim~\cite{barvinok-pommersheim}, projection of
 lattice-point sets by Barvinok and Woods~\cite{barvinok-woods}, and conversion
 between rational generating functions and explicit piecewise counting
 functions by Verdoolaege and Woods~\cite{verdoolaege-woods}.  Explicit
 formulas for rational polygons based on Dedekind--Rademacher sums are given by
 Beck and Robins~\cite{beck-robins-rational-polygons}; see also their monograph
 \cite{beck-robins}.  Those results seek short encodings or evaluate a fixed
 polyhedral counting problem.  Here the output itself is the dense vector of
 $N$ consecutive coefficients, so our analysis instead charges the dense
 polynomial work needed to produce all of them.

 Primitive lattice points in planar domains have a substantial analytic and
 algorithmic literature.  Huxley and Nowak study asymptotics in convex planar
 domains~\cite{huxley-nowak}, while Pawlewicz and P\u{a}tra\c{s}cu give sublinear
 algorithms for Farey order statistics and primitive points in polygons
 \cite{pawlewicz-patrascu}.  Reusing the same weighted direction--side
contributions across every endpoint, while returning the complete prefix,
distinguishes our task from a single primitive-point query.  Rational
 staircases also occur in the theory of Dedekind--Carlitz polynomials and
 rational cones~\cite{beck-haase-matthews}; Breuer and von Heymann give a
 Euclidean-algorithm recursion for rational staircases in $\mathbb Z^2$
 \cite{breuer-heymann}.  Our construction instead uses an index region with product
 denominators, and its recursion is balanced and dyadic so that local
 polynomial spans, rather than Euclidean remainders, control the cost.

 On the algebraic side, the product- and remainder-tree tools used below go
 back to Borodin and Moenck~\cite{borodin-moenck} and are treated systematically
 in~\cite{modern-computer-algebra}; Newton inversion of formal series uses the
 classical fast-series machinery of Brent and Kung~\cite{brent-kung}.

For each fixed leading direction coordinate, geometric series eliminate the side variables and
leave two rational kernels over a wedge of direction indices.  Exchanging the
direction pair with the side pair gives a square-root cover: every contributing
tuple has a direction index or a side index at most $\sqrt N$.
Inclusion--exclusion reduces the
whole truncated series to a small wedge and one triangular overlap.  A balanced
recursion decomposes each resulting index region into dyadic rectangles and stores every
partial rational sum over the product denominator of its own index intervals.
This locality prevents intermediate polynomials from acquiring the full global
degree at every recursive node.  Coefficientwise M\"obius inversion then
restores primitive directions, and five ordinary prefix sums recover every
output value.

 The resulting ring-level theorem is valid over a commutative ring with identity
 in which $6$ is invertible.  Every formal-series denominator created by the
 recursion has constant term one.  Over $\mathbb Z$, the same
 construction remains integral and the boundary polynomial $F_0(n)$ in
 \eqref{av:eq:F0} is evaluated by an exact division by $6$ before it is
 added.  With $\M(N)=O(N\log N)$ multiplication the bound becomes
 $O(N\log^2 N)$.  The implementation evaluates the identities in two
 transform-friendly 61-bit fields and reconstructs the exact output only after
 the modular series computations.

Our contributions are the square-root cover at the series level, explicit
closed overlap formulas, the local-denominator recursion over the two index regions and its
span analysis, exact primitive recovery, and an NTT/CRT implementation compared
 with the $O(N^{3/2})$ all-values algorithm of~\cite{paper}.  The same prior
work also studies the single-query problem.  A companion paper develops
that direction further, giving an exact $O(n\log n)$-operation algorithm for
one prescribed value $F(n)$ by M\"obius divisor layers and batched Euclidean
floor-moment evaluation~\cite{companion-one-value}; the present paper instead
targets the complete prefix by rational generating series.

The next section gives the complete geometric and algebraic overview.  The
following sections derive the kernels, square-root cover, overlap formulas,
dyadic-region evaluator, complexity bound, primitive recovery, and complete
pseudocode.  Experiments follow the theorem, while detailed derivations,
executable formulas, exact arithmetic, and bit complexity are placed in the
appendices.

\section{Overview of the all-values method}

The all-values construction computes the whole prefix
$F(1),\ldots,F(N)$ simultaneously.  This section introduces the two coefficient
arrays and follows them through the square-root cover, local-denominator
recursion, primitive extraction, and final prefix sums.  The organizing idea is
to group contributions by their larger bounding-box dimension instead of
recomputing the same direction and side sums separately for every endpoint.

After the geometric symmetries are fixed, $a>b>0$ are the positive coordinates
of a primitive direction vector and $x\geqslant y>0$ are the integer
multipliers of the two perpendicular sides.

\begin{definition}[Threshold and bounding-box gap]\label{av:def:threshold-gap}
For such a representative, its
\emph{threshold} and \emph{bounding-box gap} are, respectively,
\[
 m:=ax+by,
 \qquad
 g:=(ax+by)-(ay+bx)=(a-b)(x-y).
\]
Thus $m$ is the larger bounding-box dimension and $g\geqslant0$ is the
difference between the two dimensions.
\end{definition}
The placement count in every grid larger than the threshold is a quadratic
polynomial in $n-m$.  We therefore build two coefficient arrays: one counts
tuples of threshold $m$, and the other records their bounding-box gaps.  Once
these arrays are known, five ordinary prefix sums recover every value $F(n)$.

The coefficient arrays are obtained as follows.  For one fixed direction
coordinate $a$, geometric series sum the two side variables explicitly; the
remaining direction coordinate becomes $\ell=a+b$ with
$a<\ell<2a$.  The layers therefore form one wedge-shaped index region.
We call the involution $(a,b)\leftrightarrow(x,y)$ \emph{direction--side
symmetry}.  It implies that every relevant tuple has either a
direction index or a side index at most $\sqrt N$.  Inclusion--exclusion then
reduces the full series to the small part of this wedge and one triangular
overlap correction, with only $O(\sqrt N)$ indices in each role.  A balanced
recursion covers these regions by dyadic rectangles.  A rectangle wholly
inside a region is evaluated at once from two interval product trees; only
rectangles meeting a boundary are split.  Each recursive state keeps the
denominator belonging to its own two intervals, which prevents the
intermediate polynomials from acquiring the full global degree too early.
Finally, coefficientwise M\"obius inversion removes nonprimitive directions,
and the five prefix sums produce the output table.

We now make this outline exact.  We use $z$ as the formal generating variable
and truncate every series modulo $z^{N+1}$.

\begin{definition}[Staircase]\label{av:def:staircase}
An \emph{affine staircase} is a finite set of integer pairs obtained by
intersecting a Cartesian product of integer intervals with one or two strict
affine half-planes.  In this paper, \emph{staircase} always means an affine
staircase whose defining slopes and number of boundary lines are independent
of~$N$.
\end{definition}

The following standard parametrization fixes the geometric multiplicities used
throughout the paper and matches the notation of~\cite{paper}.

\begin{lemma}[Geometric parametrization]\label{av:lem:parametrization}
Define
\begin{equation}
 F_0(n)=\binom n2^2+\sum_{s=2}^{n-1}(s-1)(n-s)^2
       =\frac{n(n-1)^2(2n-1)}6.
 \label{av:eq:F0}
\end{equation}
For $x\geqslant y\geqslant1$, put
\[
 \mult(x,y)=
 \begin{cases}
  2,&x=y,\\
  4,&x>y.
 \end{cases}
\]
Then the contribution of the strict direction representatives is
\begin{equation}
 F_1(n)=
 \sum_{\substack{a>b>0,\ x\geqslant y\geqslant1\\
                   \gcd(a,b)=1,\ ax+by\leqslant n}}
 \mult(x,y)(n-ax-by)(n-ay-bx),
 \label{av:eq:F1}
\end{equation}
and the complete count is
\begin{equation}
 F(n)=F_0(n)+F_1(n).
 \label{av:eq:split}
\end{equation}
\end{lemma}

\begin{proof}
Every rectangle has perpendicular side vectors $x(a,b)$ and $y(-b,a)$, where
$x,y$ are positive integers and $(a,b)$ is primitive.  Reflections and
interchange of the sides give a unique representative with
$a\geqslant b\geqslant0$ and $x\geqslant y$.  Its bounding-box dimensions are
$ax+by$ and $ay+bx$, so it has
$(n-ax-by)(n-ay-bx)$ translations.  For $a>b>0$, the reflection orbit has size
two when $x=y$ and four when $x>y$, giving~\eqref{av:eq:F1}.

There are two boundary direction classes.  If $b=0$, primitivity forces
$(a,b)=(1,0)$ and gives the $\binom n2^2$ axis-parallel rectangles.  If $a=b$,
primitivity forces $(a,b)=(1,1)$.  Writing $s=x+y$, this class contributes
$\sum_{s=2}^{n-1}(s-1)(n-s)^2$.  This proves~\eqref{av:eq:F0}; the three
classes are disjoint and exhaustive, proving~\eqref{av:eq:split}.
\end{proof}

\begin{samepage}
\begin{theorem}[All-values theorem]\label{av:thm:main}
 Let $R$ be a commutative ring with identity in which $6$ is a unit, and let
 $\M$ satisfy the regularity assumptions stated in the introduction.  The
 images in $R$ of $F(1),\ldots,F(N)$ can be computed in
 $O(\M(N)\log N)$ ring operations and $O(N\log N)$ ring elements of working
 memory.  Over $\mathbb Z$, the same algebraic bounds hold when the cost is
 counted in integer-ring operations and stored integers: all staircase
 denominators have constant term one, and the sole division by $6$ in
 $F_0(n)$ is exact.  If $\M(N)=O(N\log N)$, the arithmetic bound becomes
 $O(N\log^2N)$.
\end{theorem}
\end{samepage}

\begin{remark}[Modular realization]
The theorem is independent of a particular multiplication implementation.
The modular NTT realization uses the transform-friendly field hypothesis in
Appendix~\ref{av:sec:beyond} and reconstructs only the final coefficients.
\end{remark}

The next sections prove the stages above in that order.  The longer algebraic
derivations and implementation refinements are collected in
Appendices~\ref{app:all-values-details} and~\ref{app:all-values-arithmetic}.

\section{One direction layer and the two generating kernels}

This section eliminates the two side variables for one leading direction
coordinate and identifies the two rational series used by the rest of the
algorithm.  The first records how many tuples attain each threshold; the second
records the gap between their two bounding-box dimensions.

\begin{definition}[Direction layer and primitive-free strict kernels]
\label{av:def:direction-layer}
For a fixed $a\geqslant2$, the sum over directions $1\leqslant b<a$ and strict
side representatives $x>y$ is the \emph{direction layer} indexed by~$a$.  Put
\begin{equation}
 f_j=1-z^j,\qquad
 \mathsf{A}_j=\frac{z^j}{f_j},\qquad
 \mathsf{B}_j=\frac{z^j}{f_j^2}.
 \label{av:eq:atoms}
\end{equation}
The \emph{primitive-free strict kernels} are
\begin{align}
 H_0(z)&=\sum_{a>b\geqslant1}\mathsf{A}_a\mathsf{A}_{a+b},
 \label{av:eq:H0}\\
 C_0(z)&=\sum_{a>b\geqslant1}(a-b)\mathsf{B}_a\mathsf{A}_{a+b}.
 \label{av:eq:C0}
\end{align}
Here \emph{primitive-free} means that $\gcd(a,b)=1$ has not yet been imposed;
\emph{strict} refers to the representatives $a>b$ and $x>y$.
\end{definition}

\begin{lemma}[Layer generating functions]\label{av:lem:layer-series}
For a fixed layer $a$, put $\ell=a+b$ and $t=x-y$.  Its threshold and gap are
\[
 m=at+\ell y,\qquad (a-b)(x-y)=(2a-\ell)t,
\]
and its two generating series are
\begin{align}
 \mathcal H_a(z)
   &=\mathsf A_a\sum_{a<\ell<2a}\mathsf A_\ell,
 \label{av:eq:Hlayer}\\
 \mathcal C_a(z)
   &=\mathsf B_a\sum_{a<\ell<2a}(2a-\ell)\mathsf A_\ell.
 \label{av:eq:Clayer}
\end{align}
Moreover,
\begin{align}
 \coeff{m}{H_0}
 &=\#\{a>b>0,\ x>y>0: ax+by=m\},
 \label{av:eq:Hcone}\\
 \coeff{m}{C_0}
 &=\sum_{\substack{a>b>0,\ x>y>0\\ax+by=m}}
   (a-b)(x-y).
 \label{av:eq:Ccone}
\end{align}
\end{lemma}

\begin{proof}
The substitution above is a bijection between $1\leqslant b<a$, $x>y>0$ and
$a<\ell<2a$, $t,y>0$.  Since
$\mathsf A_a=\sum_{t\geqslant1}z^{at}$ and
$\mathsf B_a=\sum_{t\geqslant1}t z^{at}$, summing first over $t$ and $y$ gives
the two layer formulas.  Summing the layers over $a\geqslant2$ gives
$H_0$ and $C_0$, and coefficient extraction gives the displayed
interpretations.
\end{proof}

All generating-series identities below are taken modulo $z^{N+1}$.

\section{A square-root cover of the full series}

This section applies direction--side symmetry before coefficient extraction,
ensuring that every retained rational term has a small direction index or a
small side index.  Its purpose is to replace the unbounded direction wedge by
two copies with only $K=\floor{\sqrt N}$ active indices, correcting their
intersection explicitly.

Set
\[
 K=\floor{\sqrt N}.
\]
If $ax+by\leqslant N$, then $ax\leqslant N$; hence $a>K$ and $x>K$ cannot both hold.
Moreover, the involution
\[
 (a,b)\longleftrightarrow(x,y)
\]
preserves both the exponent and the weight in~\eqref{av:eq:Ccone}.

\begin{definition}[Small and overlap kernels]\label{av:def:small-overlap}
The \emph{small kernels} $H_{\mathrm{s}}(z)$ and $C_{\mathrm{s}}(z)$ are the
restrictions of $H_0(z)$ and $C_0(z)$, respectively, to tuples with
$a\leqslant K$.  The \emph{overlap kernels} $H_\cap(z)$ and $C_\cap(z)$ are
their further restrictions to tuples satisfying both $a\leqslant K$ and
$x\leqslant K$.
\end{definition}

\begin{lemma}[Square-root cover]\label{av:lem:cover}
Modulo $z^{N+1}$,
\begin{equation}
 H_0=2H_{\mathrm{s}}-H_{\cap},\qquad
 C_0=2C_{\mathrm{s}}-C_{\cap}.
 \label{av:eq:cover}
\end{equation}
\end{lemma}

\begin{proof}
The admissible tuples are covered by $\{a\leqslant K\}\cup\{x\leqslant K\}$.
The two sets have equal generating series by the involution above, and their
intersection must be subtracted once.
\end{proof}

Writing $\ell=a+b$, the small parts are
\begin{align}
 H_{\mathrm{s}}
 &=\sum_{2\leqslant a\leqslant K}\mathcal H_a
  =\sum_{2\leqslant a\leqslant K}\ \sum_{a<\ell<2a}
   \mathsf{A}_a\mathsf{A}_\ell,
 \label{av:eq:Hsmall}\\
 C_{\mathrm{s}}
 &=\sum_{2\leqslant a\leqslant K}\mathcal C_a
  =\sum_{2\leqslant a\leqslant K}\ \sum_{a<\ell<2a}
   (2a-\ell)\mathsf{B}_a\mathsf{A}_\ell.
 \label{av:eq:Csmall}
\end{align}
Only factors $f_j$ with $j<2K$ occur.

\section{Closed formulas for the overlap}

The square-root cover counts its intersection twice.  This section evaluates
that overlap explicitly and reduces both kernels to a constant-size family of
separable rational atoms over two staircase-shaped index regions.

Define the finite series
\begin{align}
 W_j&=\sum_{y=1}^{K-1}z^{jy}
     =\frac{z^j-z^{Kj}}{f_j},
 \label{av:eq:W}\\
 R_j&=\sum_{y=1}^{K-1}(K-y)z^{jy}
     =\frac{(K-1)z^j-Kz^{2j}+z^{(K+1)j}}{f_j^2},
 \label{av:eq:R}
\end{align}
and
\begin{equation}
 Q^{(1)}_a=\frac{z^{(K+1)a}}{f_a},\qquad
 Q^{(2)}_a=\frac{z^{(K+1)a}}{f_a^2}.
 \label{av:eq:Q}
\end{equation}

For fixed $a>b$ and $\ell=a+b$, the overlap condition is
$t,y\geqslant1$ and $t+y\leqslant K$.  Summing this finite triangle gives
\begin{align}
 H_\cap
 &=\sum_{2\leqslant a\leqslant K}\sum_{1\leqslant b<a}
   \left(\mathsf{A}_aW_{a+b}-Q^{(1)}_aW_b\right),
 \label{av:eq:Hcap}\\
 C_\cap
 &=\sum_{2\leqslant a\leqslant K}\sum_{1\leqslant b<a}(a-b)
   \left(\mathsf{B}_aW_{a+b}-Q^{(1)}_aR_b-Q^{(2)}_aW_b\right).
 \label{av:eq:Ccap}
\end{align}
The two finite-sum identities behind these formulas are derived in
Appendix~\ref{app:overlap-derivation}.
The terms indexed by $a+b$ are evaluated over $\mathcal R_{\wedge}$ and retain
the inequalities $a<\ell<2a$ from the layer formula.  The correction terms
indexed by $b$ are evaluated over $\mathcal R_{\triangle}$ with
$1\leqslant b<a$.

Every factor in~\eqref{av:eq:Hsmall}, \eqref{av:eq:Csmall},
\eqref{av:eq:Hcap}, and~\eqref{av:eq:Ccap} is a linear combination of a constant
number of atoms
\begin{equation}
 \frac{\chi(j)z^{\alpha j}}{(1-z^j)^r},
 \qquad r\in\{1,2\},\quad \alpha\in\{1,\ldots,K+1\},
 \label{av:eq:genericatom}
\end{equation}
where $\chi(j)$ is a polynomial of degree at most one.

\section{Rectangle decomposition and local denominators}

The layer formula and the overlap correction leave exactly two index
regions.  The computational core splits each region into dyadic rectangles.
A balanced recursion follows the boundary of each staircase while storing
every partial rational sum over the denominator of its own local intervals.
The denominator construction is a specialized product tree; classical
product and remainder trees and their complexity are treated
in~\cite{borodin-moenck,modern-computer-algebra}.

We now evaluate sums of products of~\eqref{av:eq:genericatom} over either of the
two staircase regions
\begin{align}
 \mathcal R_{\wedge}&=\{(i,j):2\leqslant i\leqslant K,\ i<j<2i\},
 \label{av:eq:wedgeregion}\\
 \mathcal R_{\triangle}&=\{(i,j):2\leqslant i\leqslant K,\ 1\leqslant j<i\}.
 \label{av:eq:triangleregion}
\end{align}
In $\mathcal R_{\wedge}$ the indices are $i=a$ and $j=a+b$; in
$\mathcal R_{\triangle}$ they are $i=a$ and $j=b$.  Thus both regions encode
the same direction sums written in the coordinates used by their rational
atoms.

\begin{definition}[Staircase-sum operators]\label{av:def:staircase-operators}
For two families
\[
U_i=\frac{\chi(i)z^{\alpha i}}{f_i^r},\qquad
V_j=\frac{\psi(j)z^{\beta j}}{f_j^s},
\]
where $\chi$ and $\psi$ are polynomials of degree at most one and
$r,s,\alpha,\beta$ are integers satisfying
$1\leqslant r,s\leqslant2$ and $1\leqslant\alpha,\beta\leqslant K+1$, define
\begin{align}
 \operatorname{Wedge}(U,V)
   &:=\sum_{(i,j)\in\mathcal R_{\wedge}}U_iV_j
     =\sum_{i=2}^{K}\ \sum_{i<j<2i}U_iV_j,
 \label{av:eq:wedgeoperator}\\
\operatorname{Triangle}(U,V)
   &:=\sum_{(i,j)\in\mathcal R_{\triangle}}U_iV_j
     =\sum_{i=2}^{K}\ \sum_{1\leqslant j<i}U_iV_j.
 \label{av:eq:triangleoperator}
\end{align}
\end{definition}
Both sums are evaluated by the same interval-pair routine described below.
The triangular call has the single boundary $j=i$, while the wedge call tests
both $j=i$ and $j=2i$.
At the root the interval pairs are $I=[2,K]$, $J=[3,2K-1]$ for the wedge and
$I=[2,K]$, $J=[1,K-1]$ for the triangle; empty ranges are omitted.
The rectangle decomposition is useful because every complete block separates:
for intervals $I,J$ with $I\times J$ inside the region,
\[
 \sum_{i\in I}\sum_{j\in J}U_iV_j
 =\left(\sum_{i\in I}U_i\right)
  \left(\sum_{j\in J}V_j\right).
\]
The interval trees below represent these two one-dimensional sums over their
local common denominators.

\begin{definition}[Marked interval data]
\label{av:def:marked-data}
For an interval $I$ in the first index role, define
\begin{equation}
 D_I=\prod_{i\in I}f_i^r,
 \qquad
 P_I=\sum_{i\in I}\chi(i)z^{\alpha i}\frac{D_I}{f_i^r}.
 \label{av:eq:marked}
\end{equation}
Here $D_I$ is the \emph{local denominator} and $P_I$ is its
\emph{marked numerator}; thus $P_I/D_I=\sum_{i\in I}U_i$.
For an interval $J$ in the second role, define analogously
\[
 E_J=\prod_{j\in J}f_j^s,
 \qquad
 T_J=\sum_{j\in J}\psi(j)z^{\beta j}\frac{E_J}{f_j^s}.
\]
Thus $E_J$ and $T_J$ are the local denominator and marked numerator for the
second role, and $T_J/E_J=\sum_{j\in J}V_j$.
\end{definition}

\begin{definition}[Stored dense span]\label{av:def:dense-span}
Every stored polynomial carries the smallest monomial shift known from its
construction.  Its \emph{stored dense span} is the length of the remaining
coefficient interval; the zero polynomial has span zero.  Leading zero
coefficients created by a later cancellation need not be removed.
\end{definition}

The marked data are constructed in balanced interval trees.  If
$I=I_0\sqcup I_1$, then
\begin{equation}
 D_I=D_{I_0}D_{I_1},\qquad
 P_I=P_{I_0}D_{I_1}+D_{I_0}P_{I_1}.
 \label{av:eq:markedmerge}
\end{equation}
The analogous merge constructs $E_J,T_J$.  All products are truncated modulo
$z^{N+1}$.

Figure~\ref{av:fig:staircases} illustrates the geometric test made at every
recursive state.  A product of dyadic index intervals is accepted, discarded,
or split according to its position relative to the straight boundary lines.
The four block polynomials $D_I,E_J,P_I,T_J$ used by this test are the marked
interval data of Definition~\ref{av:def:marked-data}.

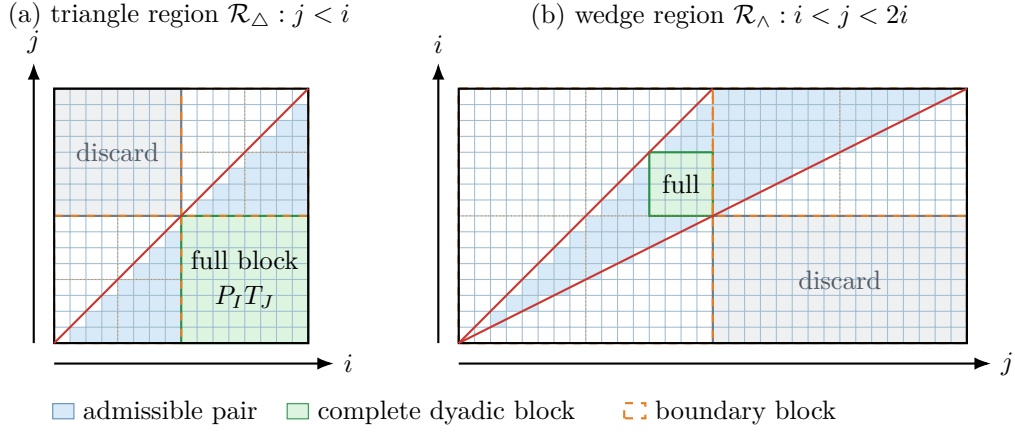
\begin{figure}[H]
  \centering
  \begin{tikzpicture}[
    x=1cm,y=1cm,
    every node/.style={font=\normalsize,align=center},
    stairaxis/.style={-{Latex[length=2.2mm]},thick}
  ]
    \node at (2.35,5.62) {(a) triangle region $\mathcal R_\triangle: j<i$};
    \node at (9.50,5.62) {(b) wedge region $\mathcal R_\wedge: i<j<2i$};

    \begin{scope}[shift={(0.70,1.30)}]
      \foreach \i in {0,...,15}{
        \foreach \j in {0,...,15}{
          \ifnum\j<\i
            \fill[staircellblue] ({0.21*\i},{0.21*\j}) rectangle
              ({0.21*(\i+1)},{0.21*(\j+1)});
          \fi
        }
      }
      \fill[stairblockgreen] (1.68,0) rectangle (3.36,1.68);
      \fill[stairblockgray] (0,1.68) rectangle (1.68,3.36);
      \draw[stairlineblue!55,xstep=0.21cm,ystep=0.21cm]
        (0,0) grid (3.36,3.36);
      \draw[stairgreenline,thick] (1.68,0) rectangle (3.36,1.68);
      \node at (2.52,0.84) {full block\\$P_I T_J$};
      \draw[stairgrayline,thick] (0,1.68) rectangle (1.68,3.36);
      \node[text=stairtextgray] at (0.84,2.52) {discard};
      \draw[stairwarmline,dashed,thick] (0,0) rectangle (1.68,1.68);
      \draw[stairwarmline,dashed,thick] (1.68,1.68) rectangle (3.36,3.36);
      \draw[stairwarmline,densely dotted] (0.84,0)--(0.84,1.68);
      \draw[stairwarmline,densely dotted] (0,0.84)--(1.68,0.84);
      \draw[stairwarmline,densely dotted] (2.52,1.68)--(2.52,3.36);
      \draw[stairwarmline,densely dotted] (1.68,2.52)--(3.36,2.52);
      \draw[thick] (0,0) rectangle (3.36,3.36);
      \draw[stairredline,thick] (0,0)--(3.36,3.36);
      \draw[stairaxis] (0,-0.27)--(3.70,-0.27) node[right] {$i$};
      \draw[stairaxis] (-0.27,0)--(-0.27,3.70) node[above] {$j$};
    \end{scope}

    \begin{scope}[shift={(6.05,1.30)}]
      \foreach \i in {0,...,15}{
        \foreach \j in {0,...,31}{
          \pgfmathtruncatemacro{\inside}{(\j>\i && \j<2*\i+1) ? 1 : 0}
          \ifnum\inside=1
            \fill[staircellblue] ({0.21*\j},{0.21*\i}) rectangle
              ({0.21*(\j+1)},{0.21*(\i+1)});
          \fi
        }
      }
      \fill[stairblockgreen] (2.52,1.68) rectangle (3.36,2.52);
      \fill[stairblockgray] (3.36,0) rectangle (6.72,1.68);
      \draw[stairlineblue!55,xstep=0.21cm,ystep=0.21cm]
        (0,0) grid (6.72,3.36);
      \draw[stairgreenline,thick] (2.52,1.68) rectangle (3.36,2.52);
      \node at (2.94,2.10) {full};
      \draw[stairgrayline,thick] (3.36,0) rectangle (6.72,1.68);
      \node[text=stairtextgray] at (5.04,0.84) {discard};
      \draw[stairwarmline,dashed,thick] (0,0) rectangle (6.72,3.36);
      \draw[stairwarmline,dashed,thick] (0,0) rectangle (3.36,3.36);
      \draw[stairwarmline,dashed,thick] (3.36,1.68) rectangle (6.72,3.36);
      \draw[stairwarmline,densely dotted] (1.68,0)--(1.68,3.36);
      \draw[stairwarmline,densely dotted] (0,1.68)--(3.36,1.68);
      \draw[stairwarmline,densely dotted] (5.04,1.68)--(5.04,3.36);
      \draw[thick] (0,0) rectangle (6.72,3.36);
      \draw[stairredline,thick] (0,0)--(3.36,3.36);
      \draw[stairredline,thick] (0,0)--(6.72,3.36);
      \draw[stairaxis] (0,-0.27)--(7.04,-0.27) node[right] {$j$};
      \draw[stairaxis] (-0.27,0)--(-0.27,3.70) node[above] {$i$};
    \end{scope}

    \node[anchor=west] at (0.55,0.38)
      {\tikz\fill[staircellblue,draw=stairlineblue]
       (0,0) rectangle (0.28,0.20);\ admissible pair};
    \node[anchor=west] at (3.65,0.38)
      {\tikz\fill[stairblockgreen,draw=stairgreenline]
       (0,0) rectangle (0.28,0.20);\ complete dyadic block};
    \node[anchor=west] at (8.10,0.38)
      {\tikz\draw[stairwarmline,dashed,thick]
       (0,0) rectangle (0.28,0.20);\ boundary block};
  \end{tikzpicture}
\caption{The two staircase regions on grids with the same square cell size:
$16\times16$ for the triangle and $16\times32$ for the wedge.  The wedge has
twice as many columns because its horizontal index ranges up to $2K$.  Blue
cells use the actual one-based indices, so each panel contains exactly
$\sum_{i=2}^{16}(i-1)=120$ admissible pairs.  Each green dyadic block
lies wholly in its region and is evaluated by one product $P_I T_J$; the gray
block is disjoint.  Orange dashed dyadic blocks meet a red boundary and are
recursively split.}
  \label{av:fig:staircases}
\end{figure}

\begin{definition}[Cross numerator]\label{av:def:cross}
Fix one of the two staircase regions~$\mathcal R$.  The two index roles are
distinct copies, even when their numerical intervals overlap.  For an interval
state $(I,J)$, its \emph{cross numerator} is
\[
 \operatorname{Cross}_{\mathcal R}(I,J)
 :=D_IE_J\!\sum_{(i,j)\in(I\times J)\cap\mathcal R}U_iV_j.
\]
When $\mathcal R$ is fixed by the call, we write
$\operatorname{Cross}(I,J)$.  Its common denominator is $D_IE_J$ even when the
underlying numerical intervals overlap.  The state is \emph{full} if
$I\times J\subseteq\mathcal R$, \emph{disjoint} if
$I\times J\cap\mathcal R=\varnothing$, and \emph{partial} otherwise.
\end{definition}

\begin{itemize}
\item If $I\times J$ is wholly contained in the staircase, then
\[
 \operatorname{Cross}(I,J)=P_I T_J.
\]
\item If it is disjoint from the staircase, the numerator is zero.
\item Otherwise split the longer interval (breaking ties arbitrarily), so
the dyadic rectangles retain bounded aspect ratio.  For
$I=I_0\sqcup I_1$,
\begin{equation}
 \operatorname{Cross}(I,J)
 =\operatorname{Cross}(I_0,J)D_{I_1}
  +\operatorname{Cross}(I_1,J)D_{I_0}.
 \label{av:eq:crossmerge}
\end{equation}
The analogous formula applies when $J$ is split.
\end{itemize}

The full/disjoint tests use only interval endpoints.  For $I=[a,b]$ and
$J=[c,d]$, they are
\begin{equation}
\begin{array}{c|c|c}
\mathcal R & I\times J\subseteq\mathcal R
           & I\times J\cap\mathcal R=\varnothing\\ \hline
\mathcal R_{\triangle} & a>d & b\leqslant c\\
\mathcal R_{\wedge} & b<c\ \text{and}\ d<2a
                     & a\geqslant d\ \text{or}\ c\geqslant2b
\end{array}
\label{av:eq:rectangle-endpoint-tests}
\end{equation}
Every other rectangle meets a boundary and is split.  These strict
inequalities implement $j<i$ and $i<j<2i$ exactly, including boundary cells.

The weights in~\eqref{av:eq:Csmall} and~\eqref{av:eq:Ccap} are absorbed by
the identities
\begin{align}
 \sum_{i<j<2i}(2i-j)U_iV_j
 &=2\operatorname{Wedge}(iU_i,V_j)
   -\operatorname{Wedge}(U_i,jV_j),
 \label{av:eq:wedgeweight}\\
 \sum_{i>j}(i-j)U_iV_j
 &=\operatorname{Triangle}(iU_i,V_j)
   -\operatorname{Triangle}(U_i,jV_j).
 \label{av:eq:triangleweight}
\end{align}

Equations~\eqref{av:eq:Hsmall}--\eqref{av:eq:Ccap} now give a complete
reduction to the two staircase operators.  Every summand is a separable
product $U_iV_j$, and the only nonseparable factors, $2i-j$ and $i-j$, are
removed by~\eqref{av:eq:wedgeweight} and~\eqref{av:eq:triangleweight}.  The
complete fixed expansion is
\begin{align}
 H_{\mathrm{s}}&=\operatorname{Wedge}(\mathsf A_i,\mathsf A_j),\nonumber\\
 C_{\mathrm{s}}&=2\operatorname{Wedge}(i\mathsf B_i,\mathsf A_j)
             -\operatorname{Wedge}(\mathsf B_i,j\mathsf A_j),\nonumber\\
 H_\cap&=\operatorname{Wedge}(\mathsf A_i,W_j)
          -\operatorname{Triangle}(Q^{(1)}_i,W_j),\nonumber\\
 C_\cap&=2\operatorname{Wedge}(i\mathsf B_i,W_j)
          -\operatorname{Wedge}(\mathsf B_i,jW_j)\nonumber\\
 &\quad-\operatorname{Triangle}(iQ^{(1)}_i,R_j)
       +\operatorname{Triangle}(Q^{(1)}_i,jR_j)\nonumber\\
 &\quad-\operatorname{Triangle}(iQ^{(2)}_i,W_j)
       +\operatorname{Triangle}(Q^{(2)}_i,jW_j).
 \label{av:eq:fixed-expansion}
\end{align}
Thus all four series require only a fixed number of staircase calls.  As
written, $H_\cap$ contains one triangular call and $C_\cap$ contains four.
The implementation evaluates these five terms with two recursive triangular
traversals; Appendix~\ref{app:all-values-optimized-schedule} gives the exact
algebraic reduction.

\section{Complexity analysis}\label{av:sec:complexity}

Having reduced both kernels to a constant number of staircase calls, this
section charges their polynomial work scale by scale.  The key point is that a
boundary meets few rectangles at a coarse scale, while the local polynomial
span of each such rectangle grows in the opposite proportion.

\begin{lemma}[Local span]\label{av:lem:span}
For an interval $I\subseteq[1,2K]$, the marked numerator $P_I$ of
Definition~\ref{av:def:marked-data} has dense span $O(K|I|)$.  A cross
numerator $\operatorname{Cross}_{\mathcal R}(I,J)$ has dense span
$O(K(|I|+|J|))$; in particular, a balanced state whose interval lengths are
$\Theta(h)$ has span $O(Kh)$.
\end{lemma}

\begin{proof}
For one role, the denominator degree is at most
$2\sum_{i\in I}i=O(K|I|)$.  After extracting the monomial shift supplied by
the construction, the spread of the exponents $\alpha i$ is
$O(K|I|)$.  For a cross state, the corresponding bounds for the two index
roles add, giving $O(K(|I|+|J|))$.
\end{proof}

\begin{lemma}[Boundary states at one scale]\label{av:lem:boundary-states}
Fix a dyadic scale $h$.  In either staircase recursion, the number of partial
states whose two interval lengths are $\Theta(h)$ is $O(K/h)$.  The number of
full states returned at the same scale is also $O(K/h)$.
\end{lemma}

\begin{proof}
Every interval occurring in the recursion is a node of one of the two balanced
dyadic interval trees.  The root side lengths differ by at most a constant
factor, and splitting a longest side preserves this property: up to endpoint
rounding and the constant-size base cases, every partial state satisfies
$c\leqslant |I|/|J|\leqslant C$ for absolute constants $c,C>0$.  Thus a state
whose first interval has length in $[h,2h)$ uses only $O(1)$ adjacent dyadic
levels for its second interval.

At any fixed level the first-role intervals have disjoint interiors and
partition a range of length $O(K)$, so only $O(K/h)$ of them have length
$\Theta(h)$.  Fix one such interval $I$.  Over the vertical strip
$I\times\mathbb R$, either boundary $j=i$ or $j=2i$ spans a second-coordinate
interval of length $O(h)$.  At each of the $O(1)$ admissible second-role levels,
only $O(1)$ intervals $J$ with disjoint interiors can meet that span.  Hence
each fixed $I$ participates in $O(1)$ partial states, and the total number of
partial states is $O(K/h)$.

A full state is returned without descendants.  Apart from a possible full
root in a constant-size base case, every returned full state is a child of a
partial state at the same scale up to a fixed factor.  Since a split has two
children, the same $O(K/h)$ bound holds for full states.
\end{proof}

\begin{lemma}[Total span]\label{av:lem:mass}
The sum of all local polynomial spans visited by one staircase computation is
$O(K^2\log K)=O(N\log N)$.
\end{lemma}

\begin{proof}
At scale $h$, Lemma~\ref{av:lem:boundary-states} gives $O(K/h)$ partial or
full cross states that construct a polynomial; disjoint states return before
doing so.  Lemma~\ref{av:lem:span} bounds every such span by $O(Kh)$, for
$O(K^2)$ total cross span at that scale.  A marked interval tree likewise has
$O(K/h)$ nodes of length $\Theta(h)$, each with span $O(Kh)$.  Summing over
$O(\log K)$ dyadic scales proves the claim.
\end{proof}

\begin{proposition}[Staircase evaluation complexity]
\label{av:prop:staircase-complexity}
One call to $\operatorname{Wedge}$ or $\operatorname{Triangle}$, including its
marked denominator trees and root inversion, uses
$O(\M(N)\log N)$ arithmetic operations and $O(N\log N)$ coefficient-ring
elements of working memory.
\end{proposition}
\begin{proof}
At one scale, Lemmas~\ref{av:lem:span}
and~\ref{av:lem:boundary-states} bound the sum of operand spans by $O(K^2)$,
and every recursion state performs only a constant number of polynomial
products and additions.  Superadditivity and monotonicity of $\M$,
together with stability under constant dilation, bound all products at that
scale by $O(\M(K^2))=O(\M(N))$; additions cost $O(K^2)$.  The
$O(\log K)$ scales therefore cost $O(\M(N)\log N)$.

The root denominators have degree
$O(\bigl(\sum_{j\leqslant2K}j\bigr))=O(K^2)=O(N)$ and constant term one.
Newton inversion consequently costs $O(\M(N))$
\cite{brent-kung,modern-computer-algebra}.  Retaining all marked-tree and
cross-state polynomials uses the total-span bound
$O(K^2\log K)=O(N\log N)$; depth-first evaluation cannot increase it.  A
constant number of aligned calls may share a marked tree or final inverse
without changing either asymptotic bound.
\end{proof}

\section{Primitive extraction and recovery of \texorpdfstring{$F$}{F}}

The kernels above still include nonprimitive directions.  This section
applies M\"obius inversion coefficientwise and then recovers the whole table
$F(1),\ldots,F(N)$ from ordinary prefix sums; we use the standard divisor-sum
form of M\"obius inversion~\cite{apostol}.

For $1\leqslant m\leqslant N$, write $h_m=\coeff{m}{H_0}$ and
$c_m=\coeff{m}{C_0}$, and set all zeroth entries to zero.  M\"obius inversion
on the common divisor of the direction gives
\begin{equation}
 h_m^*=\sum_{d\mid m}\mu(d)h_{m/d},\qquad
 c_m^*=\sum_{d\mid m}\mu(d)d\,c_{m/d}.
 \label{av:eq:mobius}
\end{equation}
Here $\mu$ is the M\"obius function.  The factor $d$ in the second formula
occurs because the gap weight scales
linearly with the direction.

For the side diagonal $x=y$, primitive extraction has the closed form
\begin{equation}
 d_m^*=\floor{\frac{m-1}{2}}.
 \label{av:eq:diagonal-closed}
\end{equation}
Indeed, for a fixed sum $s\geqslant3$ there are $\varphi(s)/2$ primitive pairs
$a>b>0$ with $a+b=s$, where $\varphi$ is Euler's totient function.  Together
with $\sum_{s\mid m}\varphi(s)=m$, the missing cases $s=1,2$ give exactly the
floor in~\eqref{av:eq:diagonal-closed}.

\begin{definition}[Primitive threshold arrays]\label{av:def:threshold-arrays}
Put
\begin{equation}
 X_m=4h_m^*+2d_m^*,\qquad Y_m=4c_m^*.
 \label{av:eq:thresholdarrays}
\end{equation}
The arrays $(X_m)_{m\leqslant N}$ and $(Y_m)_{m\leqslant N}$ are the
\emph{primitive threshold arrays}: $X_m$ supplies the quadratic coefficient
and $Y_m$ the linear gap coefficient for threshold~$m$.
\end{definition}
Then the strict-direction contribution is
\begin{equation}
 F_1(n)=
 \sum_{1\leqslant m<n}\bigl(X_m(n-m)^2+Y_m(n-m)\bigr),
 \qquad F(n)=F_0(n)+F_1(n).
 \label{av:eq:recover}
\end{equation}
All values of~\eqref{av:eq:recover} are obtained in linear time from the five
prefix sums
\[
 \sum X_m,\quad\sum mX_m,\quad\sum m^2X_m,
 \quad\sum Y_m,\quad\sum mY_m.
\]
For a direct streaming implementation, define
\[
 S_{rX}(n)=\sum_{m<n}m^rX_m\quad(0\leqslant r\leqslant2),
 \qquad
 S_{rY}(n)=\sum_{m<n}m^rY_m\quad(0\leqslant r\leqslant1).
\]
Expanding the two powers in~\eqref{av:eq:recover} gives the exact update
formula used by the implementation:
\begin{equation}
 F_1(n)=n^2S_{0X}(n)-2nS_{1X}(n)+S_{2X}(n)
 +nS_{0Y}(n)-S_{1Y}(n),
 \qquad F(n)=F_0(n)+F_1(n).
 \label{av:eq:streaming-recover}
\end{equation}
After emitting $F(n)$, update
$S_{rX}\gets S_{rX}+n^rX_n$ and
$S_{rY}\gets S_{rY}+n^rY_n$.  Thus the sums used for the next endpoint still
contain precisely the indices $m<n+1$.

\begin{proposition}[Exact primitive recovery]\label{av:prop:recovery}
Equations~\eqref{av:eq:mobius}--\eqref{av:eq:recover} recover exactly the
geometric count~\eqref{av:eq:F1} for every $1\leqslant n\leqslant N$.
The two M\"obius transforms take $O(N\log N)$ arithmetic operations by direct
divisor--multiple loops, and the final recovery of all $N$ values takes
$O(N)$ operations.
\end{proposition}
\begin{proof}[Proof sketch]
Coefficientwise M\"obius inversion removes the common divisor of $(a,b)$; the
gap kernel gains the displayed factor $d$ because its weight is homogeneous
of degree one in the direction.  For a primitive tuple, the two box
dimensions differ by $(a-b)(x-y)$, which gives $F_1(n)$ in
\eqref{av:eq:recover}; adding $F_0(n)$ gives the complete count.
Divisor--multiple loops cost $O(N\log N)$ and the five prefixes recover each
$F(n)$ in constant time.  Full details are in
Appendix~\ref{app:primitive-recovery}.
\end{proof}

\begin{proof}[Proof of Theorem~\ref{av:thm:main}]
Lemma~\ref{av:lem:cover} is an exact inclusion--exclusion identity, and
equations~\eqref{av:eq:Hsmall}--\eqref{av:eq:Ccap} give its four rational
series.  Equations~\eqref{av:eq:wedgeweight} and
\eqref{av:eq:triangleweight} rewrite them as a constant number of
$\operatorname{Wedge}$ and $\operatorname{Triangle}$ calls.
Proposition~\ref{av:prop:staircase-complexity} evaluates those calls in
$O(\M(N)\log N)$ operations and $O(N\log N)$ ring elements.  Proposition
 \ref{av:prop:recovery} then extracts primitive directions and all output
 values in $O(N\log N)$ additional operations and $O(N)$ additional ring
 elements.
Since polynomial multiplication is at least linear, this is within the
claimed arithmetic bound.

Over $\mathbb Z$, every series operation above is integral because each
inverted denominator has constant term one.  Evaluate
$F_0(n)=n(n-1)^2(2n-1)/6$ by exact integer division before adding it in
\eqref{av:eq:recover}; this proves the integer variant without assuming that
$6$ is a unit.  Under the modular hypothesis of
Appendix~\ref{av:sec:beyond}, the same identities are evaluated in each field
and the final residues are uniquely reconstructed because the modulus product
exceeds $F(N)$.  Sequential modular runs preserve the stated peak
ring-element memory bound.
\end{proof}

\section{Complete algorithm}\label{av:sec:algorithm}

The recursion~\eqref{av:eq:markedmerge}--\eqref{av:eq:crossmerge} is the
complete evaluator for one wedge or triangular sum.  The following
top-level pseudocode shows how those sums are assembled into the output
table.  It is written over one coefficient ring; a modular implementation
repeats it in the required fields and reconstructs only the final
coefficients.

\begin{algorithm}[H]
  \small
  \caption{Complete all-values algorithm.}
  \label{av:alg:all-values}
  \begin{algorithmic}
    \AlgLine{0}{\textbf{procedure} $\textsc{AllRectangleCounts}(N)$}
    \AlgLine{1}{$K\gets\floor{\sqrt N}$; truncate every series modulo $z^{N+1}$}
    \AlgLine{1}{form the atom families $\mathsf A,\mathsf B,W,R,Q^{(1)},Q^{(2)}$}
    \AlgLine{1}{evaluate $H_{\mathrm{s}},C_{\mathrm{s}},H_\cap,C_\cap$ by
      \eqref{av:eq:Hsmall}--\eqref{av:eq:Ccap} using the staircase recursion}
    \AlgLine{1}{$H_0\gets2H_{\mathrm{s}}-H_\cap$;
      $C_0\gets2C_{\mathrm{s}}-C_\cap$}
    \AlgLine{1}{extract $h_m=[z^m]H_0$ and $c_m=[z^m]C_0$}
    \AlgLine{1}{compute $h_m^*,c_m^*$ by the two M\"obius transforms
      in~\eqref{av:eq:mobius}}
    \AlgLine{1}{form $X_m,Y_m$ by~\eqref{av:eq:thresholdarrays} and their five
      prefix sums}
    \AlgLine{1}{\textbf{for} $n=1,\ldots,N$ \textbf{do} recover $F(n)$ from
      \eqref{av:eq:recover}}
    \AlgLine{1}{\textbf{return} $(F(1),\ldots,F(N))$}
    \AlgLine{0}{\textbf{end procedure}}
  \end{algorithmic}
\end{algorithm}

For completeness, the two lower-level algorithms used inside
Algorithm~\ref{av:alg:all-values} are given next.  For
$\mathcal R=\mathcal R_{\wedge}$ the roots are $I=[2,K]$ and
$J=[3,2K-1]$; for $\mathcal R=\mathcal R_{\triangle}$ they are $I=[2,K]$
and $J=[1,K-1]$.  Empty endpoint ranges are omitted.

\begin{algorithm}[H]
  \small
  \caption{Local-denominator evaluation of one staircase sum.}
  \label{av:alg:staircase}
  \begin{algorithmic}
    \AlgLine{0}{\textbf{procedure} $\textsc{Staircase}(U,V,\mathcal R)$}
    \AlgLine{1}{choose $(I_{\rm root},J_{\rm root})$ from the two ranges above}
    \AlgLine{1}{build the balanced marked trees $(D_I,P_I)$ and $(E_J,T_J)$}
    \AlgLine{1}{$S\gets\textsc{Cross}(I_{\rm root},J_{\rm root},\mathcal R)$}
    \AlgLine{1}{\textbf{return}
      $S/(D_{I_{\rm root}}E_{J_{\rm root}})\pmod{z^{N+1}}$}
    \AlgLine{0}{\textbf{end procedure}}
    \AlgLine{0}{\textbf{procedure} $\textsc{Cross}(I,J,\mathcal R)$}
    \AlgLine{1}{\textbf{if} $I\times J$ is disjoint from $\mathcal R$
      \textbf{then return} $0$}
    \AlgLine{1}{\textbf{if} $I\times J\subseteq\mathcal R$
      \textbf{then return} $P_I T_J$}
    \AlgLine{1}{\textbf{if} $|I|\geqslant|J|$ \textbf{then}}
    \AlgLine{2}{split $I=I_0\sqcup I_1$ into its balanced children}
    \AlgLine{2}{\textbf{return}
      $\textsc{Cross}(I_0,J,\mathcal R)D_{I_1}
      +\textsc{Cross}(I_1,J,\mathcal R)D_{I_0}$}
    \AlgLine{1}{\textbf{else}}
    \AlgLine{2}{split $J=J_0\sqcup J_1$ into its balanced children}
    \AlgLine{2}{\textbf{return}
      $\textsc{Cross}(I,J_0,\mathcal R)E_{J_1}
      +\textsc{Cross}(I,J_1,\mathcal R)E_{J_0}$}
    \AlgLine{1}{\textbf{end if}}
    \AlgLine{0}{\textbf{end procedure}}
  \end{algorithmic}
\end{algorithm}

The call uses $\mathcal R_{\wedge}$ for $\operatorname{Wedge}$ and
$\mathcal R_{\triangle}$ for $\operatorname{Triangle}$.  Multiplying an atom
family by its index supplies the weighted variants in
\eqref{av:eq:wedgeweight} and~\eqref{av:eq:triangleweight}.  The optimized
schedule uses two forms of fusion.  \emph{Rank-two fusion} carries two marked
numerators through the same denominator-lifting recursion.  \emph{Frontier
fusion} evaluates same-region families separately down to a common fixed-depth
frontier, adds their equal-denominator numerators there, and shares all lifts
above the frontier.  The implementation combines these operations with the
two-traversal reduction of Appendix~\ref{app:all-values-optimized-schedule} and
the shared common-denominator construction of
Appendix~\ref{av:sec:reductions}.

\begin{algorithm}[H]
  \small
  \caption{Optimized ring-level schedule for the main text.}
  \label{av:alg:optimized-schedule}
  \begin{algorithmic}
    \AlgLine{0}{\textbf{procedure} $\textsc{OptimizedKernels}(N)$}
    \AlgLine{1}{$K\gets\floor{\sqrt N}$; truncate every series modulo $z^{N+1}$}
    \AlgLine{1}{build one denominator tree on $[1,2K]$ and reuse its left half
      for the first index role}
    \AlgLine{1}{construct index-weighted marked numerators by the logarithmic-
      derivative identity~\eqref{av:eq:log-derivative-marking}}
    \AlgLine{1}{evaluate wedge terms by rank-two fusion and frontier fusion}
    \AlgLine{1}{replace the five triangular traversals by the two traversals
      in~\eqref{av:eq:five-to-two} plus~\eqref{av:eq:closed-weighted-triangle}}
    \AlgLine{1}{lift $H_0,C_0$ to one denominator and invert it once}
    \AlgLine{1}{\textbf{return} the coefficient arrays $(h_m),(c_m)$}
    \AlgLine{0}{\textbf{end procedure}}
  \end{algorithmic}
\end{algorithm}

Appendix~\ref{app:all-values-optimized-schedule} derives the two-traversal
reduction used by Algorithm~\ref{av:alg:optimized-schedule}.
Appendix~\ref{app:all-values-arithmetic} gives the shared-tree schedule, exact
modular realization, CRT reconstruction, and coefficient-growth bounds.

\section{Experiments}\label{sec:all-values-experiments}

The experiment compares two single-threaded all-values algorithms: the
$O(N^{3/2})$ method of~\cite{paper} and the $O(N\log^2 N)$ method developed
here.  Both methods construct the complete prefix $F(1),\ldots,F(N)$.  For a fair all-values
baseline, the threshold sums of~\cite{paper} are accumulated into every prefix
endpoint by range additions; this adds $O(N\log N)$ work and does not change
its $O(N^{3/2})$ bound.  During timing, endpoint-only output suppresses all but
the last printed value; it does not suppress construction of the prefix.

For the measured range, the baseline stores exact values in signed 128-bit
integers.  The proposed
program implements Algorithms~\ref{av:alg:all-values}--\ref{av:alg:optimized-schedule}
and the formulas in Appendices~\ref{app:all-values-executable-formulas}
and~\ref{av:sec:exact-arithmetic}.  Its two modular computations use
the primes in~\eqref{av:eq:implementation-primes} sequentially, followed by
incremental CRT.  Direct convolution is used when the shorter operand has at
most 48 coefficients, and
families sharing a staircase recursion are combined after four common levels;
these two tuning constants are fixed for every input size.

The programs were run on an AMD Ryzen 9 5950X system with 128 GiB of physical
memory.
GCC 16.1.0 from MSYS2 UCRT64 compiled both as C++20 with
\texttt{-O3}, \texttt{-DNDEBUG}, and \texttt{-march=native}.  The active power
scheme was \texttt{High performance}.  Each process ran with priority
class \texttt{High} and processor-affinity mask \texttt{0x4}, which selects one
logical processor.  Measurements cover $N=2^{10},\ldots,2^{23}$.  Each plotted
point is the arithmetic mean of ten runs.  The ancillary data record the means
used in the figure; the plot is a descriptive scaling comparison, and no
statistical inference from the timing sample is intended.

Figure~\ref{fig:all-values-timing} shows that the two running-time curves meet
near $N=2^{21}$, after which the proposed implementation overtakes the
baseline.  When divided by $N\log_2^2N$, the proposed running time remains
approximately stable across the measured range.  By contrast, the baseline
time divided by $N^{3/2}$ increases at the larger sizes.  This trend is
consistent with the baseline's increasingly sparse updates across a growing
array, although no hardware-counter measurements were made.  The proposed
method has a larger constant at small sizes; its dominant NTT passes use dense,
regular access patterns, which plausibly contributes to the steadier measured
throughput as the problem grows.

\begin{figure}[H]
  \centering
  \resizebox{\textwidth}{!}{%
  \begin{tikzpicture}
  \begin{groupplot}[
    group style={group size=2 by 1,horizontal sep=1.75cm},
    width=6.05cm,height=4.85cm,
    xmin=9.5,xmax=23.5,
    xtick={10,11,12,13,14,15,16,17,18,19,20,21,22,23},
    xticklabels={$2^{10}$,$2^{11}$,$2^{12}$,$2^{13}$,$2^{14}$,
                 $2^{15}$,$2^{16}$,$2^{17}$,$2^{18}$,$2^{19}$,
                 $2^{20}$,$2^{21}$,$2^{22}$,$2^{23}$},
    xlabel={maximum size $N=2^k$},
    grid=major,grid style={benchmarkgrid!65},
    axis line style={black!80},
    x tick label style={font=\scriptsize,rotate=45,anchor=east},
    y tick label style={font=\scriptsize},
    label style={font=\footnotesize},
    xlabel style={yshift=2pt},
    title style={font=\footnotesize,
                 at={(axis description cs:0.5,1.005)},anchor=south},
    legend style={font=\scriptsize,draw=benchmarkgrid,fill=white,
                  fill opacity=0.94,text opacity=1,cells={anchor=west},
                  inner xsep=3pt,inner ysep=2pt}
  ]
    \nextgroupplot[
      title={Running times of the algorithms},
      ylabel={running time (seconds)},
      ymode=log,ymin=0.001,ymax=20000,
      ytick={0.001,0.01,0.1,1,10,100,1000,10000},
      legend pos=south east
    ]
      \addplot[benchmarkbaseline,thick,mark=triangle*,mark size=2.3pt]
        coordinates {
          (10,0.002214)(11,0.005048)(12,0.015111)(13,0.045226)
          (14,0.125147)(15,0.379615)(16,1.089356)(17,3.941845)
          (18,10.224035)(19,36.820932)(20,126.022173)
          (21,447.077930)(22,1340.156035)(23,4165.479790)
        };
      \addlegendentry{$O(N^{3/2})$}
      \addplot[benchmarkproposed,thick,mark=square*,mark size=2.1pt]
        coordinates {
          (10,0.048258)(11,0.118476)(12,0.273383)(13,0.632560)
          (14,1.412468)(15,3.478712)(16,7.603546)(17,17.318354)
          (18,38.028357)(19,86.722955)(20,191.880799)
          (21,428.397063)(22,924.307940)(23,2112.573910)
        };
      \addlegendentry{$O(N\log_2^2 N)$}

    \nextgroupplot[
      title={Normalized running times},
      ylabel={normalized running time ($\mu$s)},
      ymode=log,ymin=0.04,ymax=0.6,
      ytick={0.05,0.1,0.2,0.5},
      yticklabels={$0.05$,$0.1$,$0.2$,$0.5$},
      legend style={at={(axis description cs:0.98,0.72)},anchor=east}
    ]
      \addplot[benchmarkbaseline,thick,mark=triangle*,mark size=2.3pt]
        coordinates {
          (10,0.067565918)(11,0.054465867)(12,0.057643890)
          (13,0.060996268)(14,0.059674740)(15,0.063998304)
          (16,0.064930677)(17,0.083068172)(18,0.076174997)
          (19,0.096992890)(20,0.117367295)(21,0.147210358)
          (22,0.156014696)(23,0.171447115)
        };
      \addlegendentry{$O(N^{3/2})$}
      \addplot[benchmarkproposed,thick,mark=square*,mark size=2.1pt]
        coordinates {
          (10,0.471269531)(11,0.478095945)(12,0.463499281)
          (13,0.456904124)(14,0.439847985)(15,0.471830512)
          (16,0.453206658)(17,0.457192233)(18,0.447736634)
          (19,0.458201966)(20,0.457479475)(21,0.463210064)
          (22,0.455314433)(23,0.476065094)
        };
      \addlegendentry{$O(N\log_2^2 N)$}
  \end{groupplot}
  \end{tikzpicture}%
  }
  \caption{Mean running times over ten runs under the common one-processor
  launch policy.  Left: elapsed time on a logarithmic scale.  Right: the baseline
  time normalized by $N^{3/2}$ and the proposed time normalized by
  $N\log_2^2N$; both normalized values are in microseconds.  Only the arithmetic
  means are shown; the figure is intended as a descriptive scaling comparison.}
  \label{fig:all-values-timing}
\end{figure}
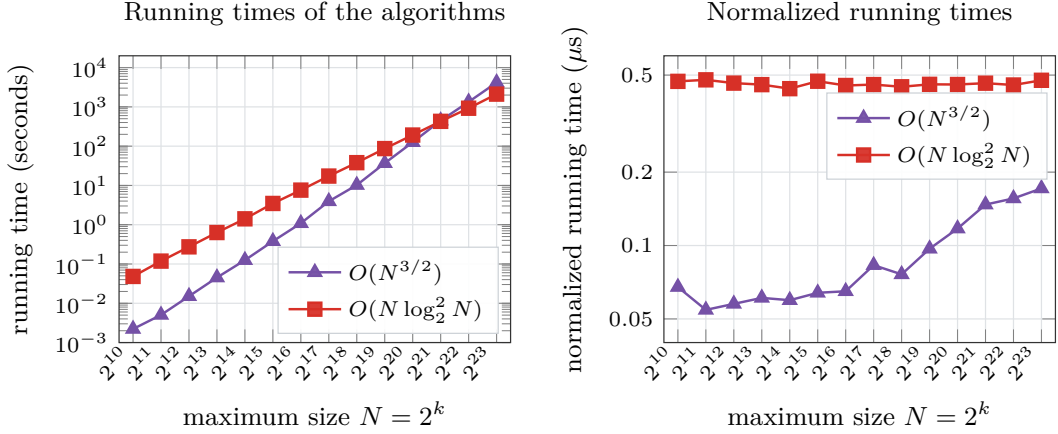

\section{Reproducibility}

The reference implementations and build and benchmark instructions are
available in the
\href{https://github.com/flykiller/lattice-rectangles-all-values}
{\nolinkurl{flykiller/lattice-rectangles-all-values}} repository on GitHub.

\section{Future work}\label{sec:future-work}

Two extensions appear especially useful: parallel evaluation that preserves
shared denominator data, and block-separable production of the output prefix.

The two modular runs, the butterflies within one NTT stage, and independent
staircase states at a fixed
recursion scale offer immediate parallelism.  A useful many-core
implementation should nevertheless preserve the near-linear total work and
the $O(N\log N)$ shared-memory bound: assigning whole recursion subtrees to
workers without coordination can replicate large denominator polynomials and
make memory, rather than arithmetic, the limiting resource.

A more ambitious goal is to make the output block-separable after one shared
preprocessing phase.  Partition $[1,N]$ into disjoint intervals
$[k_r,\ell_r]$ and let worker $r$ compute
$F(k_r),\ldots,F(\ell_r)$ independently.  Such preprocessing would have to
turn the marked denominator trees and the common series inverse into a
representation supporting coefficient-window queries, for example through
middle products or transposed polynomial multiplication.  It would also have
to expose the input coefficients needed by the divisor--multiple M\"obius
transform without repeating the full transform on every worker.  The final
recovery step is already compatible with this organization: the five
prefix-sum accumulators evaluated immediately before $k_r$ seed a local scan
of~\eqref{av:eq:recover} over the whole block.  The unresolved question is
whether the preceding kernel construction and primitive extraction can be
given the same block interface with total work $O(\M(N)\log N)$, shared space
$O(N\log N)$, and running time close to the total work divided by the number
of processors, up to polylogarithmic overhead.

\section{Conclusion}

The complete prefix of lattice-rectangle counts can be computed exactly in
$O(\M(N)\log N)$ coefficient-ring operations and $O(N\log N)$ ring elements
of working memory.  The square-root cover is applied before coefficient
extraction, and the remaining rational staircase sums are evaluated over local
product denominators.  Primitive extraction and five prefix sums then recover
all values.  Under transform-friendly quasi-linear multiplication the bound
specializes to $O(N\log^2 N)$.  The measured normalization is consistent with
that scaling: the running-time curves meet near $N=2^{21}$ and the modular
implementation is faster at the larger measured sizes while retaining
certified exact reconstruction.

\clearpage
\appendix

\section{Rational-series details}
\label{app:all-values-details}

This appendix supplies the algebraic derivations behind the optimized execution
schedule, executable formulas, and the detailed recovery proof omitted from
the main line of the argument.  The central polynomial-span proof is in
Section~\ref{av:sec:complexity}.

\subsection{Derivation of the overlap formulas}
\label{app:overlap-derivation}

This subsection verifies the two finite triangular sums used in the overlap
correction, using standard finite geometric and arithmetic--geometric
identities~\cite{concrete-math}, and shows that every resulting factor has the
generic atom form in~\eqref{av:eq:genericatom}.

Fix $a>b$, write $\ell=a+b$, and recall that the overlap condition is
$t,y\geqslant1$ and $t+y\leqslant K$.  Summing first over $t$ gives
\begin{equation}
 \sum_{\substack{t,y\geqslant1\\t+y\leqslant K}}z^{at+\ell y}
 =\mathsf{A}_aW_\ell-Q^{(1)}_aW_b.
 \label{av:eq:overlapH}
\end{equation}
Applying $u\,d/du$ to the finite sum with $u^t$ in place of $z^{at}$ and then
setting $u=z^a$, or evaluating the arithmetic--geometric sum directly, gives
\begin{equation}
 \sum_{\substack{t,y\geqslant1\\t+y\leqslant K}}tz^{at+\ell y}
 =\mathsf{B}_aW_\ell-Q^{(1)}_aR_b-Q^{(2)}_aW_b.
 \label{av:eq:overlapC}
\end{equation}
Summing the first identity over $2\leqslant a\leqslant K$,
$1\leqslant b<a$ gives~\eqref{av:eq:Hcap}; multiplying the second identity by
$a-b$ before the same summation gives~\eqref{av:eq:Ccap}.  Together with
\eqref{av:eq:Hsmall} and~\eqref{av:eq:Csmall}, every factor is a linear
combination of a constant number of atoms
\[
 \frac{\chi(j)z^{\alpha j}}{(1-z^j)^r},
 \qquad r\in\{1,2\},\quad \alpha\in\{1,\ldots,K+1\},
\]
where $\chi(j)$ has degree at most one.  Therefore the overlap is a constant
number of wedge and triangular staircase sums.

\subsection{Optimized operator schedule and the two-traversal reduction}
\label{app:all-values-optimized-schedule}

The fixed expansion~\eqref{av:eq:fixed-expansion} is an identity of rational
series.  This subsection turns it into the optimized schedule used by the
implementation.  In the terminology introduced before
Algorithm~\ref{av:alg:optimized-schedule}, rank-two fusion carries the two
marked numerators belonging to one gap weight, while frontier fusion shares
the upper lifts of several same-region recursions.  The weighted
overlap nevertheless contains five literal triangular traversals after the
monomials of $W$ and $R$ are expanded.  The rest of this subsection reduces
those five traversals to two before invoking the recursion.

For
\[
 \operatorname{Triangle}_{\mathrm{gap}}(U,V)
 :=\sum_{1\leqslant j<i\leqslant K}(i-j)U_iV_j
\]
we extend every first-role family that starts at $i=2$ by setting its
$i=1$ term to zero.  Thus all families in the identities below have the same
index set $\{1,\ldots,K\}$, all unqualified index sums in this subsection use
that set, and the exact transposition identity is
\begin{equation}
 \operatorname{Triangle}_{\mathrm{gap}}(U,V)
 =\operatorname{Triangle}_{\mathrm{gap}}(V,U)
 +\left(\sum_i iU_i\right)\left(\sum_jV_j\right)
 -\left(\sum_iU_i\right)\left(\sum_j jV_j\right).
 \label{av:eq:triangle-transpose}
\end{equation}
Indeed, the two triangular sums partition the off-diagonal part of the
separable double sum on the right, while the diagonal has zero weight.

Write
\[
 \mathsf A_j^{(K)}=\frac{z^{Kj}}{f_j},\qquad
 R_j^{(12)}=\frac{(K-1)z^j-Kz^{2j}}{f_j^2},\qquad
 R_j^{(3)}=\frac{z^{(K+1)j}}{f_j^2},
\]
so that $W=\mathsf A-\mathsf A^{(K)}$ and
$R=R^{(12)}+R^{(3)}$.  Also put
\[
 \operatorname{Sep}(U,V)
 :=\left(\sum_i iU_i\right)\left(\sum_jV_j\right)
   -\left(\sum_iU_i\right)\left(\sum_j jV_j\right).
\]
After transposition, the last atom $R^{(3)}$ and the factor
$-\mathsf A^{(K)}$ combine.  Because the transposed $Q^{(1)}$ factor is active
only for the original first index $j\geqslant2$, their combined second factor
is
\begin{equation}
 \mathsf V_j=
 \begin{cases}
  -z^K/f_1,&j=1,\\
  \bigl(-z^{Kj}+z^{(K+1)j}\bigr)/f_j=-z^{Kj},
     &2\leqslant j\leqslant K.
 \end{cases}
 \label{av:eq:closed-triangle-factor}
\end{equation}
Equivalently,
\[
 \mathsf V_j=-z^{Kj}-\mathbf 1_{j=1}\frac{z^{K+1}}{f_1}.
\]
Here $\mathbf 1_{P}$ equals one when the predicate $P$ holds and zero
otherwise.
Applying~\eqref{av:eq:triangle-transpose} gives
\begin{align}
 &\operatorname{Triangle}_{\mathrm{gap}}(Q^{(1)},R)
  +\operatorname{Triangle}_{\mathrm{gap}}(Q^{(2)},W)\nonumber\\
 &\qquad=
  \operatorname{Triangle}_{\mathrm{gap}}(R^{(12)},Q^{(1)})
  +\operatorname{Triangle}_{\mathrm{gap}}(Q^{(2)},\mathsf A)\nonumber\\
 &\qquad\quad
  +\operatorname{Triangle}_{\mathrm{gap}}(Q^{(2)},\mathsf V)
  +\operatorname{Sep}(Q^{(1)},R).
 \label{av:eq:five-to-two}
\end{align}
The first two terms on the right are the two remaining recursive traversals.
The first uses one marked numerator because the shifts $z^i$ and $z^{2i}$
are adjacent.  The fourth term, $\operatorname{Sep}(Q^{(1)},R)$, uses only
root sums, while the third is evaluated without a staircase traversal.

Put $w=z^K$ and define
\[
 S_0=\sum_{i=2}^{K}Q_i^{(2)},\qquad
 S_1=\sum_{i=2}^{K}iQ_i^{(2)},\qquad
 S_w=\sum_{i=2}^{K}w^iQ_i^{(2)}.
\]
The elementary finite identity
\begin{equation}
 \sum_{j=1}^{i-1}(i-j)w^j
 =\frac{(i-1)w-iw^2+w^{i+1}}{(1-w)^2}
 \label{av:eq:finite-weighted-geometric}
\end{equation}
gives the combined triangular term in closed form:
\begin{equation}
 \operatorname{Triangle}_{\mathrm{gap}}(Q^{(2)},\mathsf V)
 =-\frac{w(S_1-S_0)-w^2S_1+wS_w}{(1-w)^2}
  -\frac{z^{K+1}}{f_1}(S_1-S_0).
 \label{av:eq:closed-weighted-triangle}
\end{equation}
The two formal divisions by $1-w=1-z^K$ are linear: coefficients in each
residue class modulo $K$ are prefix-summed, and the operation is repeated for
the square.  The exceptional $j=1$ term uses the already available factor
$f_1$ and one root product.

Consequently~\eqref{av:eq:five-to-two} replaces five literal weighted
triangular traversals by two.  The remaining work is a constant number of
marked-root constructions, ordinary root products, and linear formal
divisions.  Only adjacent shifts are packed into one dense numerator; widely
separated shifts remain shifted polynomials until the closed formula is
applied.  Thus Lemma~\ref{av:lem:span} is unchanged.  The index-one atom
exposed by the transposition is lifted along one root-to-leaf path, whose
geometrically growing products cost $O(\M(N))$ in total.  This optimization
changes neither the $O(\M(N)\log N)$ time bound nor the $O(N\log N)$ space
bound.

The direct recursive evaluator and the optimized ring-level schedule now appear
as Algorithms~\ref{av:alg:staircase} and~\ref{av:alg:optimized-schedule} in the
main text.

\subsection{Executable local-polynomial and shared-root formulas}
\label{app:all-values-executable-formulas}

This subsection makes explicit the formulas needed to translate the optimized
schedule into code.  A stored shifted polynomial is a pair
$(s,p)$ representing $z^sp(z)$, where $p(0)$ may be nonzero.  If
$u=\min\{s,t\}$, its exact truncated arithmetic is
\begin{align}
 (s,p)+(t,q)
   &=\left(u,
      z^{s-u}p+z^{t-u}q\pmod {z^{N+1-u}}\right),\nonumber\\
 (s,p)(t,q)
   &=\left(s+t,
      pq\pmod {z^{N+1-s-t}}\right).
 \label{av:eq:shifted-polynomial-arithmetic}
\end{align}
Pairs with shift greater than $N$ are zero.  Trailing zero coefficients are
removed after every operation.  At a leaf $i$, the denominator is
$(1-z^i)^r$ and the marked numerator is
$\chi(i)z^{\alpha i}$, or zero when that atom is inactive.  Equations
\eqref{av:eq:markedmerge} and~\eqref{av:eq:crossmerge} then apply without any
implicit change of denominator.
The linear divisions used in~\eqref{av:eq:closed-weighted-triangle} are
equally explicit: if $Q=P/(1-z^d)$ and missing coefficients have value zero,
then
\begin{equation}
 q_j=p_j+\mathbf 1_{j\geqslant d}q_{j-d}
 \qquad(0\leqslant j\leqslant N).
 \label{av:eq:residue-prefix-division}
\end{equation}
Applying this recurrence twice divides by $(1-z^d)^2$; the case $d=1$
is ordinary coefficient prefix summation.

The full and empty rectangle tests are the endpoint conditions in
\eqref{av:eq:rectangle-endpoint-tests}.  Every remaining rectangle is partial
and is split as in Algorithm~\ref{av:alg:staircase}.

For rank-two fusion, let $(P_I^{(u)},T_J^{(u)})$, $u\in\{1,2\}$, be two
marked-family pairs on the same denominator trees, and fix scalars
$\lambda_1,\lambda_2$.  The fused numerator is defined on a full rectangle by
\begin{equation}
 \operatorname{Cross}_{\boldsymbol\lambda}(I,J)
 =\lambda_1P_I^{(1)}T_J^{(1)}
  +\lambda_2P_I^{(2)}T_J^{(2)}.
 \label{av:eq:rank-two-full-block}
\end{equation}
On a partial rectangle it obeys exactly the component-independent lifting
recurrence~\eqref{av:eq:crossmerge}.  Thus one may either carry the two
components separately or add their transformed products before the inverse
transform; both routes return the same numerator.

Finally, write
\[
 \Delta_1=\prod_{i=1}^{K}(1-z^i),
 \qquad
 \Delta_2=\prod_{j=1}^{2K}(1-z^j).
\]
The aliased first-role tree includes an inactive leaf at $i=1$.  Cancelling
one resulting factor $f_1=1-z$ represents the two assembled kernels as
\[
 H_0=\frac{N_H}{\Delta_1\Delta_2/f_1},
 \qquad
 C_0=\frac{N_C}{\Delta_1^2\Delta_2/f_1}.
\]
Consequently the exact shared-root lift used by the implementation is
\begin{equation}
 D=\frac{\Delta_1^2\Delta_2}{f_1},
 \qquad
 H_0=\frac{\Delta_1N_H}{D},
 \qquad
 C_0=\frac{N_C}{D}.
 \label{av:eq:shared-root-denominator}
\end{equation}
No general polynomial division is required for cancelling $f_1$: division by
$1-z$ is coefficient prefix summation.  To invert $D$, start with
$G_1=D(0)^{-1}=1$ and double the correct prefix by
\begin{equation}
 G_{2m}=G_m\bigl(2-DG_m\bigr)\pmod {z^{\min\{2m,N+1\}}}.
 \label{av:eq:newton-series-inverse}
\end{equation}
Then $[z^q]H_0=[z^q](\Delta_1N_HG)$ and
$[z^q]C_0=[z^q](N_CG)$ for $0\leqslant q\leqslant N$.
Frontier fusion only changes where equal-denominator numerators are added; it
does not change any formula above or the returned series.

\subsection{Primitive extraction and exact recovery}
\label{app:primitive-recovery}

This subsection supplies the detailed multiplicity and homogeneity argument
behind the two M\"obius transforms and the final placement polynomial.

\begin{proof}[Full proof of Proposition~\ref{av:prop:recovery}]
Equation~\eqref{av:eq:Hcone} counts every strict-side tuple with threshold
$m=ax+by$.  If $(a,b)=d(a_0,b_0)$, then the threshold is multiplied by $d$;
hence $h_m=\sum_{d\mid m}h^*_{m/d}$.  The weight
$(a-b)(x-y)$ also gains a factor $d$, so
\[
 c_m=\sum_{d\mid m}d\,c^*_{m/d}.
\]
Their Dirichlet inverses are $\mu$ and $d\mapsto\mu(d)d$, proving
\eqref{av:eq:mobius}.

For a primitive strict tuple let
\[
 m=ax+by,\qquad \delta=(a-b)(x-y).
\]
Its other bounding-box dimension is $ay+bx=m-\delta$, so its placement count
is
\[
 (n-m)(n-m+\delta)=(n-m)^2+\delta(n-m).
\]
Direction symmetry and exchange of the unequal side coordinates give
multiplicity four, hence $4h_m^*,4c_m^*$.  On the side diagonal the weight is two and
$\delta=0$, giving $2d_m^*$.  This proves the formula for $F_1(n)$ in
\eqref{av:eq:recover}; adding $F_0(n)$ by~\eqref{av:eq:split} gives $F(n)$.

The divisor--multiple loops have total length
$\sum_{d\leqslant N}\floor{N/d}=O(N\log N)$.  The five displayed prefixes
recover each $F(n)$ in constant time after one linear pass.
\end{proof}
\section{Exact arithmetic and implementation details}
\label{app:all-values-arithmetic}

This appendix explains how the ring-level construction is shared across calls,
how it is evaluated exactly in fixed-width modular arithmetic, and which
assumptions change beyond the supported endpoint range.

\subsection{Shared trees and fused traversals}
\label{av:sec:reductions}

The fixed expansion~\eqref{av:eq:fixed-expansion} contains terms with the
same denominators, regions, or weights.  The following reductions combine
them before the staircase recursion while preserving the recursion states and
the bound of Proposition~\ref{av:prop:staircase-complexity}.

Put
\[
 f_1=1-z,\qquad \Delta=\prod_{i=2}^{K}(1-z^i).
\]
For denominator exponent one, the left half of the second-role tree already
has denominator
$E_{[1,K]}=f_1\Delta$.  Represent the first role on the same interval tree
$[1,K]$, set its marked numerator at leaf $i=1$ to zero, and let a
lightweight structural copy refer to the existing denominator polynomials.
This represents the same rational sum over $i\geqslant2$, with one extra
common factor $f_1$.  At the root that factor is cancelled by formal division
by $1-z$, which is a prefix sum of coefficients.  The tree $E_{[1,K]}^2$
built for the squared second role is subsequently reused as the squared
first-role tree.  No dense polynomial is copied, and the extra tree structure
has only $O(K)$ nodes.

The implementation also fuses algebraically related products before their
inverse transforms.  A denominator lift has the form
$P_LD_R+P_RD_L$; both products are accumulated in the frequency domain and
share one inverse transform.  More generally, two aligned rank-one forms
$\lambda_1 U_1V_1+\lambda_2 U_2V_2$ over the same full rectangle share their inverse
transform, and a rank-two version of \textsc{Cross} carries the pair through
the same partial-state recursion.  The two weighted terms in
\eqref{av:eq:wedgeweight} or~\eqref{av:eq:triangleweight} are therefore carried
by one denominator-lifting traversal.

Several numerator families have the same region and root denominator but
widely separated shifts.  Packing them into one dense polynomial would
destroy the shifted-span bound.  The program instead descends a fixed number
of common boundary levels, evaluates each family separately below that
frontier, adds the local numerators there, and lifts the sum only once through
the shared upper levels.  This frontier fusion is used for both
regions.  The strict kernels $H_0$ and $C_0$ are finally lifted to one common
root denominator; one Newton inverse and two products with that inverse
produce both quotient series.

One frequently used marked family is obtained without a product-tree pass.
For
\[
 E_I(z)=\prod_{i\in I}(1-z^i)
\]
logarithmic differentiation gives
\begin{equation}
 \sum_{i\in I}\frac{iz^i}{1-z^i}
 =-\frac{zE_I'(z)}{E_I(z)}.
 \label{av:eq:log-derivative-marking}
\end{equation}
Hence its marked numerator is read from the already available denominator in
linear time at every node.  The exact two-traversal reduction is in
Appendix~\ref{app:all-values-optimized-schedule}, and its ring-level execution
is summarized by Algorithm~\ref{av:alg:optimized-schedule}.

\subsection{Exact modular implementation}
\label{av:sec:exact-arithmetic}

The mathematical algorithm is independent of finite-range engineering
choices.  The reference C++ implementation uses the truncated polynomial
products, Newton inversion, M\"obius sums, and five-prefix recovery proved in
the main text.  For $N\leqslant2^{30}$ it specializes the coefficient fields,
output accumulator, and polynomial indices to fixed widths.

The common-denominator coefficients can be much larger than the final answer,
so they are never materialized as unbounded integers.  The complete series
construction, M\"obius extraction, and linear recovery are run independently
modulo
\begin{equation}
 p_1=2305842979148922881,
 \qquad
 p_2=2305842811645198337.
 \label{av:eq:implementation-primes}
\end{equation}
In both fields $6$ is invertible, as required by
Theorem~\ref{av:thm:main}.  Both primes satisfy $2^{32}\mid p_i-1$, and the $2$-primary component
generated by $3$ has order $2^{32}$.  An ordinary truncated product needs
transform length at most $4N\leqslant2^{32}$, so all required roots of unity
exist.

For completeness, let $R=2^{64}$ and
$p'=-p^{-1}\pmod R$.  A field element $x$ is stored as $xR\bmod p$, and the
Montgomery reduction used after a $128$-bit product is
\begin{equation}
 \operatorname{Red}_p(t)
 =\frac{t+\bigl((tp')\bmod R\bigr)p}{R}\pmod p.
 \label{av:eq:montgomery-reduction}
\end{equation}
The numerator is divisible by $R$; one conditional subtraction returns the
canonical representative.  For a transform length $L\mid(p-1)$ put
$\omega_L=3^{(p-1)/L}\pmod p$.  With
\[
 \widehat a_j=\sum_{k=0}^{L-1}a_k\omega_L^{jk},
 \qquad
 a_k=L^{-1}\sum_{j=0}^{L-1}\widehat a_j\omega_L^{-jk},
\]
the truncated product used throughout the trees is
\begin{equation}
 (ab)_{<q}
 =\left(\operatorname{NTT}^{-1}_L
   \bigl(\operatorname{NTT}_L(a)\odot
         \operatorname{NTT}_L(b)\bigr)\right)_{<q},
 \quad
 L=2^{\lceil\log_2(\deg a+\deg b+1)\rceil}.
 \label{av:eq:ntt-convolution}
\end{equation}
Here $\odot$ denotes componentwise multiplication of the two transformed
vectors.
Operands are first shortened to the requested prefix, and the result is then
cut to $q$ coefficients.  A small direct-convolution cutoff and radix-two
versus radix-four butterflies affect running time but not this identity.

Each modular run returns the complete vector
$(F(1),\ldots,F(N))\bmod p_i$.  The program immediately incorporates that
vector into an exact array by incremental Chinese remaindering and releases
the modular workspace before starting the next field.  Thus the two large
denominator trees are never resident simultaneously.

Explicitly, suppose the current exact representative is $x$ modulo $M$ and
the next modular run returns $r$ modulo a coprime prime $p$.  The update is
\begin{equation}
 t=\bigl(r-(x\bmod p)\bigr)(M^{-1}\bmod p)\bmod p,
 \qquad
 x\gets x+Mt,
 \qquad
 M\gets Mp.
 \label{av:eq:incremental-crt}
\end{equation}
It is applied independently to every recovered value $F(n)$, beginning with
$M=1$ and $x=0$.

The exact value reported in~\cite{paper},
\begin{equation}
 F(2^{30})=
 4851850095158746095561485451592336296,
 \label{av:eq:exactmax}
\end{equation}
certifies the fixed reconstruction range.  The sequence is nondecreasing,
because an $n\times n$ grid embeds in the next grid.  Hence
\eqref{av:eq:exactmax} bounds every requested value with $n\leqslant2^{30}$.
Direct comparison gives $p_1p_2>F(2^{30})$, so the two residues determine
every answer uniquely and the program's two-limb unsigned accumulator is
sufficient.

The accompanying C++20 program contains a radix-four/two Cooley--Tukey
transform~\cite{cooley-tukey} over these two NTT fields
\cite{pollard1971,nussbaumer1982}, shifted dense polynomials, Newton series
inversion, the local-denominator recursion, divisor--multiple M\"obius
extraction, five-prefix recovery, and incremental CRT.  Modular multiplication
uses 64-bit Montgomery reduction~\cite{montgomery1985}.  The stated endpoint
is an arithmetic and indexing guarantee.  The in-memory staircase
representation still uses the $O(N\log N)$ space proved above, so a literal
endpoint run needs correspondingly large hardware or an out-of-core variant.

\subsection{Bit complexity and larger endpoints}
\label{av:sec:beyond}

The implementation bound $2^{30}$ comes from two fixed 61-bit moduli, a
two-word output integer, and 32-bit polynomial indices.
Extending the implementation requires the following changes.

\begin{table}[H]
\centering
\small
\begin{tabular}{p{0.25\textwidth}p{0.65\textwidth}}
\toprule
component & required change\\
\midrule
indices and degrees & replace \texttt{int} by \texttt{size\_t} or an
unsigned 64-bit type\\
output accumulator & replace the two-word \texttt{BigUInt} by a dynamic limb
array\\
exact reconstruction & choose the least number of fields whose characteristic
product exceeds a certified upper bound for $F(N)$\\
NTT roots & for $L=2^{\lceil\log_2(2N+1)\rceil}$ require
$L\mid(p_i-1)$ for every modulus\\
memory & if necessary, process tree levels out of core; the present
representation uses $O(N\log N)$ words\\
\bottomrule
\end{tabular}
\caption{Changes needed beyond $N=2^{30}$.}
\label{av:tab:beyond}
\end{table}

The asymptotic result of~\cite{paper} gives
\begin{equation}
 F(N)=\Theta(N^4\log N),\qquad
 \mathcal B_F(N):=\left\lceil\log_2(F(N)+1)\right\rceil
 =4\log_2N+\log_2\log N+O(1).
 \label{av:eq:outputbits}
\end{equation}
This determines the asymptotic size of the dynamic output accumulator.  For a
concrete endpoint, let $U_F(N)\geqslant F(N)$ be a certified numerical upper
bound and put
\[
 \mathcal B_U(N)=\left\lceil\log_2\bigl(U_F(N)+1\bigr)\right\rceil.
\]
The CRT range is determined by $U_F(N)$, not merely by the asymptotic bit
length in~\eqref{av:eq:outputbits}.

Let $\MZ(b)$ denote the bit complexity of multiplying two $b$-bit integers.
Table~\ref{av:tab:complexity-comparison} compares the previous and proposed
all-values methods.  The word-RAM has $\Theta(\log N)$-bit words.  The bit
bounds charge arithmetic on polynomially bounded exact values by
$O(\MZ(\log N))$ and use $\MZ(b)=O(b\log b)$
\cite{harvey-vdh}.  The proposed bit bound is conditional on the
transform-friendly field hypothesis stated below the table.

\begin{table}[H]
\centering
\small
\renewcommand{\arraystretch}{1.15}
\begin{tabular}{@{}lcc@{}}
\toprule
method & word-RAM operations & bit operations\\
\midrule
previous all-values method~\cite{paper}
  & $O(N^{3/2})$
  & $O(N^{3/2}\log N\log\log N)$\\
proposed all-values method
  & $O(N\log^2N)$
  & $O(N\log^3N\log\log N)$\\
\bottomrule
\end{tabular}
\caption{Time complexity in the machine-word and bit models.  The bit column
uses quasi-linear integer multiplication; the proposed bit bound is conditional
on the transform-friendly field hypothesis stated below.}
\label{av:tab:complexity-comparison}
\end{table}

For the proposed method, assume at every required transform length a
constant-size family of
coefficient fields with $\Theta(\log N)$-bit elements, the necessary roots of
unity, and characteristic product larger than a certified output bound.  A
field multiplication costs $O(\MZ(\log N))$ bit operations, so the exact
algorithm uses
\begin{equation}
 T_{\rm bit}(N)
 =O\!\left(N\log^2N\,\MZ(\log N)\right).
 \label{av:eq:bitcomplexity}
\end{equation}
In particular,
\[
 T_{\rm bit}(N)=
 \begin{cases}
  O(N\log^4N), & \text{schoolbook integer multiplication},\\
  O(N\log^3N\log\log N), & \text{quasi-linear integer multiplication}.
 \end{cases}
\]
The second line uses the $O(b\log b)$ integer-multiplication algorithm of
Harvey and van der Hoeven~\cite{harvey-vdh}; its classical fast-Fourier-
transform (FFT) predecessor is the Sch\"onhage--Strassen
algorithm~\cite{schoenhage-strassen}.  On a word random-access machine
(word-RAM)
with $\Theta(\log N)$-bit words, the same statement is
$O(N\log^2N)$ word operations and $O(N\log N)$ words of working memory.  The
staircase trees occupy $O(N\log^2N)$ bits, while the output table itself
contains $\Theta(N\log N)$ bits.

For a fixed family of $b$-bit prime fields, as long as the fields contain the
required roots, the exact number of runs is
\[
 r=\min\left\{s:\prod_{i=1}^{s}p_i>U_F(N)\right\}
  =\min\left\{s:\sum_{i=1}^{s}\log_2p_i>\log_2U_F(N)\right\}.
\]
For comparable primes this is approximately
$\lceil\mathcal B_U(N)/b\rceil$.  Repeating the
$O(N\log^2N)$ modular computation that many times yields
$O(N\log^3N)$ fixed-word operations as a range-dependent upper bound.  In the
asymptotic field model of~\eqref{av:eq:bitcomplexity}, the field elements have
$\Theta(\log N)$ bits and a constant-size family suffices.  Incremental CRT
keeps both the staircase storage and the growing exact output within
$O(N\log N)$ fixed-size words throughout this fixed-word regime.

The two primes in~\eqref{av:eq:implementation-primes} support transforms only
through length $2^{32}$.  Once
$L>2^{32}$, new primes with larger two-adic order, multiword prime fields, or
suitable extension fields are required in addition to the index and CRT
changes in Table~\ref{av:tab:beyond}.  None of these finite-width changes alters the algebraic
$O(\M(N)\log N)$ theorem.
\bibliography{references}

@misc{paper,
  author        = {Babichev, Dmitry and Babichev, Sergey},
  title         = {Counting All Lattice Rectangles in the Square Grid in Near-Linear Time},
  year          = {2026},
  eprint        = {2604.22456},
  archiveprefix = {arXiv},
  doi           = {10.48550/arXiv.2604.22456},
}

@misc{companion-one-value,
  author        = {Babichev, Dmitry and Shpakova, Tatiana},
  title         = {Counting Lattice Rectangles in {$O(n\log n)$} Operations},
  year          = {2026},
  eprint        = {2607.17961},
  archiveprefix = {arXiv},
  doi           = {10.48550/arXiv.2607.17961},
}

@misc{oeis,
  author = {{OEIS Foundation}},
  title  = {{A085582}: Number of Rectangles with Corners on an {$n \times n$} Grid},
  url    = {https://oeis.org/A085582},
}

@book{beck-robins,
  author    = {Beck, Matthias and Robins, Sinai},
  title     = {Computing the Continuous Discretely: Integer-Point Enumeration in Polyhedra},
  edition   = {2},
  series    = {Undergraduate Texts in Mathematics},
  publisher = {Springer},
  year      = {2015},
  doi       = {10.1007/978-1-4939-2969-6},
}

@article{barvinok,
  author  = {Barvinok, Alexander I.},
  title   = {A Polynomial Time Algorithm for Counting Integral Points in Polyhedra When the Dimension Is Fixed},
  journal = {Mathematics of Operations Research},
  volume  = {19},
  number  = {4},
  pages   = {769--779},
  year    = {1994},
  doi     = {10.1287/moor.19.4.769},
}

@incollection{barvinok-pommersheim,
  author    = {Barvinok, Alexander and Pommersheim, James E.},
  title     = {An Algorithmic Theory of Lattice Points in Polyhedra},
  booktitle = {New Perspectives in Algebraic Combinatorics},
  series    = {Mathematical Sciences Research Institute Publications},
  volume    = {38},
  pages     = {91--147},
  publisher = {Cambridge University Press},
  year      = {1999},
}

@article{barvinok-woods,
  author  = {Barvinok, Alexander and Woods, Kevin},
  title   = {Short Rational Generating Functions for Lattice Point Problems},
  journal = {Journal of the American Mathematical Society},
  volume  = {16},
  number  = {4},
  pages   = {957--979},
  year    = {2003},
  doi     = {10.1090/S0894-0347-03-00428-4},
}

@article{verdoolaege-woods,
  author  = {Verdoolaege, Sven and Woods, Kevin},
  title   = {Counting with Rational Generating Functions},
  journal = {Journal of Symbolic Computation},
  volume  = {43},
  number  = {2},
  pages   = {75--91},
  year    = {2008},
  doi     = {10.1016/j.jsc.2007.07.007},
}

@article{beck-robins-rational-polygons,
  author  = {Beck, Matthias and Robins, Sinai},
  title   = {Explicit and Efficient Formulas for the Lattice Point Count in Rational Polygons Using {Dedekind--Rademacher} Sums},
  journal = {Discrete \& Computational Geometry},
  volume  = {27},
  number  = {4},
  pages   = {443--459},
  year    = {2002},
  doi     = {10.1007/s00454-001-0082-3},
  eprint  = {math/0111329},
  archiveprefix = {arXiv},
}

@article{beck-haase-matthews,
  author  = {Beck, Matthias and Haase, Christian and Matthews, Asia R.},
  title   = {{Dedekind--Carlitz} Polynomials as Lattice-Point Enumerators in Rational Polyhedra},
  journal = {Mathematische Annalen},
  volume  = {341},
  number  = {4},
  pages   = {945--961},
  year    = {2008},
  doi     = {10.1007/s00208-008-0220-9},
}

@article{breuer-heymann,
  author  = {Breuer, Felix and von Heymann, Frederik},
  title   = {Staircases in {$\mathbb Z^2$}},
  journal = {Integers},
  volume  = {10},
  number  = {6},
  pages   = {807--847},
  year    = {2010},
  doi     = {10.1515/INTEG.2010.058},
}

@article{huxley-nowak,
  author  = {Huxley, Martin N. and Nowak, Werner G.},
  title   = {Primitive Lattice Points in Convex Planar Domains},
  journal = {Acta Arithmetica},
  volume  = {76},
  number  = {3},
  pages   = {271--283},
  year    = {1996},
  doi     = {10.4064/aa-76-3-271-283},
}

@article{pawlewicz-patrascu,
  author  = {Pawlewicz, Jakub and P{\u a}tra{\c s}cu, Mihai},
  title   = {Order Statistics in the {Farey} Sequences in Sublinear Time and Counting Primitive Lattice Points in Polygons},
  journal = {Algorithmica},
  volume  = {55},
  number  = {2},
  pages   = {271--282},
  year    = {2009},
  doi     = {10.1007/s00453-008-9221-z},
}

@book{apostol,
  author    = {Apostol, Tom M.},
  title     = {Introduction to Analytic Number Theory},
  series    = {Undergraduate Texts in Mathematics},
  publisher = {Springer},
  year      = {1976},
  doi       = {10.1007/978-1-4757-5579-4},
}

@book{concrete-math,
  author    = {Graham, Ronald L. and Knuth, Donald E. and Patashnik, Oren},
  title     = {Concrete Mathematics},
  edition   = {2},
  publisher = {Addison--Wesley},
  year      = {1994},
}

@article{brent-kung,
  author  = {Brent, Richard P. and Kung, H. T.},
  title   = {Fast Algorithms for Manipulating Formal Power Series},
  journal = {Journal of the ACM},
  volume  = {25},
  number  = {4},
  pages   = {581--595},
  year    = {1978},
  doi     = {10.1145/322092.322099},
}

@article{borodin-moenck,
  author  = {Borodin, Allan and Moenck, Robert},
  title   = {Fast Modular Transforms},
  journal = {Journal of Computer and System Sciences},
  volume  = {8},
  number  = {3},
  pages   = {366--386},
  year    = {1974},
  doi     = {10.1016/S0022-0000(74)80029-2},
}

@book{modern-computer-algebra,
  author    = {von zur Gathen, Joachim and Gerhard, J{\"u}rgen},
  title     = {Modern Computer Algebra},
  edition   = {3},
  publisher = {Cambridge University Press},
  year      = {2013},
  doi       = {10.1017/CBO9781139856065},
}

@article{cooley-tukey,
  author  = {Cooley, James W. and Tukey, John W.},
  title   = {An Algorithm for the Machine Calculation of Complex {Fourier} Series},
  journal = {Mathematics of Computation},
  volume  = {19},
  pages   = {297--301},
  year    = {1965},
  doi     = {10.1090/S0025-5718-1965-0178586-1},
}

@article{pollard1971,
  author  = {Pollard, John M.},
  title   = {The Fast {Fourier} Transform in a Finite Field},
  journal = {Mathematics of Computation},
  volume  = {25},
  number  = {114},
  pages   = {365--374},
  year    = {1971},
  doi     = {10.1090/S0025-5718-1971-0301966-0},
}

@book{nussbaumer1982,
  author    = {Nussbaumer, Henri J.},
  title     = {Fast Fourier Transform and Convolution Algorithms},
  edition   = {2},
  series    = {Springer Series in Information Sciences},
  volume    = {2},
  publisher = {Springer},
  year      = {1982},
  doi       = {10.1007/978-3-642-81897-4},
}

@article{montgomery1985,
  author  = {Montgomery, Peter L.},
  title   = {Modular Multiplication without Trial Division},
  journal = {Mathematics of Computation},
  volume  = {44},
  number  = {170},
  pages   = {519--521},
  year    = {1985},
  doi     = {10.1090/S0025-5718-1985-0777282-X},
}

@article{schoenhage-strassen,
  author  = {Sch{\"o}nhage, Arnold and Strassen, Volker},
  title   = {Schnelle Multiplikation gro{\ss}er Zahlen},
  journal = {Computing},
  volume  = {7},
  number  = {3--4},
  pages   = {281--292},
  year    = {1971},
  doi     = {10.1007/BF02242355},
}

@article{harvey-vdh,
  author  = {Harvey, David and van der Hoeven, Joris},
  title   = {Integer Multiplication in Time {$O(n \log n)$}},
  journal = {Annals of Mathematics},
  volume  = {193},
  number  = {2},
  pages   = {563--617},
  year    = {2021},
  doi     = {10.4007/annals.2021.193.2.4},
}

\end{document}